\documentclass{article}

\PassOptionsToPackage{numbers,compress}{natbib}

\usepackage[preprint]{neurips_2019}
\usepackage{xcolor}
\usepackage[utf8]{inputenc} 
\usepackage[T1]{fontenc}    
\usepackage{hyperref}       
\usepackage{url}            
\usepackage{booktabs}       
\usepackage{amsfonts}       
\usepackage{nicefrac}       
\usepackage{microtype}      
\usepackage{graphicx}
\usepackage{amsmath}
\usepackage{amssymb}
\usepackage{amsthm}
\usepackage{caption}
\usepackage{algorithm}
\usepackage{relsize}
\usepackage{subfigure}
\usepackage[title]{appendix}

\floatname{algorithm}{Algorithm}

\usepackage{hyperref}
\usepackage{xr}
 
\makeatletter
\newcommand*{\addFileDependency}[1]{
  \typeout{(#1)}
  \@addtofilelist{#1}
  \IfFileExists{#1}{}{\typeout{No file #1.}}
}
\makeatother

\newcommand{\indep}{\rotatebox[origin=c]{90}{$\models$}}
\newcommand{\nindep}{\not\!\perp\!\!\!\perp}
\DeclareMathOperator{\EX}{\mathbb{E}}
\DeclareMathOperator{\G}{\mathcal{G}}
\DeclareMathOperator{\N}{\mathcal{N}}
\newcommand{\R}{\mathbb{R}}
\newcommand{\Prob}{\mathbb{P}}

\DeclareMathOperator*{\argmin}{arg\,min}
\newtheorem{lemma}{Lemma}
\newtheorem{corollary}{Corollary}
\newtheorem{theorem}{Theorem}

\newtheorem{assumptions}{}

\newtheorem{example}{Example}[section]

\newtheorem{assumpone}{}

\newtheorem{assumptwo}{}

\title{Anchored Causal Inference \\ in the Presence of Measurement Error}

\author{%
  Basil Saeed\\
  Lab for Information \& Decision Systems\\
  Massachusetts Institute of Technology\\
  Cambridge, MA 02139\\
  \texttt{bsaeed@mit.edu} \\
   \And
  Anastasiya Belyaeva\\
  Lab for Information \& Decision Systems\\
  Massachusetts Institute of Technology\\
  Cambridge, MA 02139\\
  \texttt{belyaeva@mit.edu} \\
   \AND
  Yuhao Wang\\
  Lab for Information \& Decision Systems\\
  Massachusetts Institute of Technology\\
  Cambridge, MA 02139\\
  \texttt{yuhaow@mit.edu} \\
   \And
  Caroline Uhler\\
  Lab for Information \& Decision Systems\\
  Massachusetts Institute of Technology\\
  Cambridge, MA 02139\\
  \texttt{cuhler@mit.edu} \\
}

\begin{document}
\maketitle

\begin{abstract}
We consider the problem of learning a causal graph in the presence of measurement error. This setting is for example common in genomics, where  gene expression is corrupted through the measurement process. We develop a provably consistent procedure for estimating the causal structure in a linear Gaussian structural equation model from corrupted observations on its nodes, under a variety of measurement error models. 
We provide an estimator based on the method-of-moments, which can be used in conjunction with constraint-based causal structure discovery algorithms. 
We prove asymptotic consistency of the procedure and also discuss finite-sample considerations.
We demonstrate our method's performance through simulations and on real data, where we recover the underlying gene regulatory network from zero-inflated single-cell RNA-seq data.

\end{abstract}


\section{Introduction}
Determining causal relationships between a set of variables is a central task in causal inference with applications in many scientific fields including economics, biology and social sciences~\cite{friedman2000,pearl2003causality,robins2000marginal}. Directed acyclic graph (DAG) models, also known as Bayesian networks, are commonly used to represent the causal structure among variables. Learning a DAG from observations on the nodes is intrinsically hard~\cite{chickering2004NPhard}, and in general a DAG is only identifiable up to its Markov equivalence class~\cite{VP90}. In addition, in many applications there may be latent variables. While various algorithms have been developed to learn DAGs with latent variables~\cite{RFCI,spirtes2001anytimeFCI,spirtes2000causation,zhang2008completenessFCI}, without restrictions on the latent variables there may be infinitely many DAGs that can explain the data~\cite{richardson2002ancestral}, in which case the model is of little use for further interpretation and analysis.



Restrictions on the latent variables can improve model identifiability. In this paper, we consider the problem of causal discovery with measurement error, where each latent variable has exactly one corresponding observed variable (corrupted observation of the latent variable), which serves as its \emph{anchor},  and the goal is to infer the causal relationships among the latent variables; see Figure~\ref{fig:model}a. 
For instance in social sciences, the beliefs of people cannot be directly measured, but surveys can provide a noisy version of the latent variables, and we may be interested in inferring the causal structure among the latent beliefs.
Similarly, in biological applications measurment error needs to be taken into account, e.g.~when measuring brain signals using functional magnetic resonance (fMRI) or gene expression using RNA sequencing. 

While the method developed in this paper can be applied generally to causal inference in the presence of measurement noise, we will showcase its use on learning the underlying  gene regulatory network from single-cell RNA-seq data~\cite{Klein2015}. Such data is  known to suffer from dropout~\cite{Ziegenhain2017}, which manifests itself as false zeros due to too little starting RNA or technical noise. In single-cell RNA-seq experiments, it is estimated that such false zeros occur with a probability of 24-76\% across current state-of-the-art technologies~\cite{Ziegenhain2017}. Applying causal inference methods directly to such data may lead to biased estimates, where negatively correlated variables may appear positively correlated due to dropout; see Figure~\ref{fig:model}b. Currently, the typical approach for dealing with dropout is to first impute gene expression data~\cite{van2018MAGIC}. However, this may introduce artificial dependencies (Figure~\ref{fig:model}c) that may be detrimental for learning the correct causal model. It is therefore of great interest to develop algorithms that directly learn the causal structure among the latent variables from the corrupted data.  

\begin{figure}[t!]
	\centering
	\subfigure[Anchored Causal Model]{\includegraphics[width=.31\textwidth]{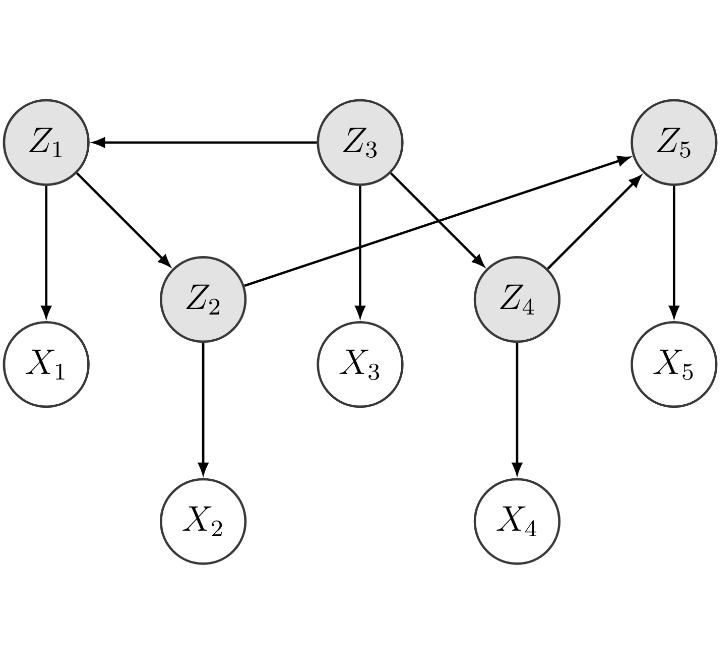}} \hspace{0.05\textwidth}
	\subfigure[Dropout]{\includegraphics[width=.15\textwidth]{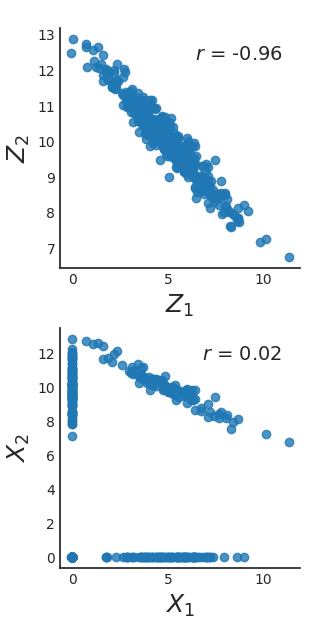}}\hspace{0.05\textwidth}
	\subfigure[Imputation]{\includegraphics[width=.29\textwidth]{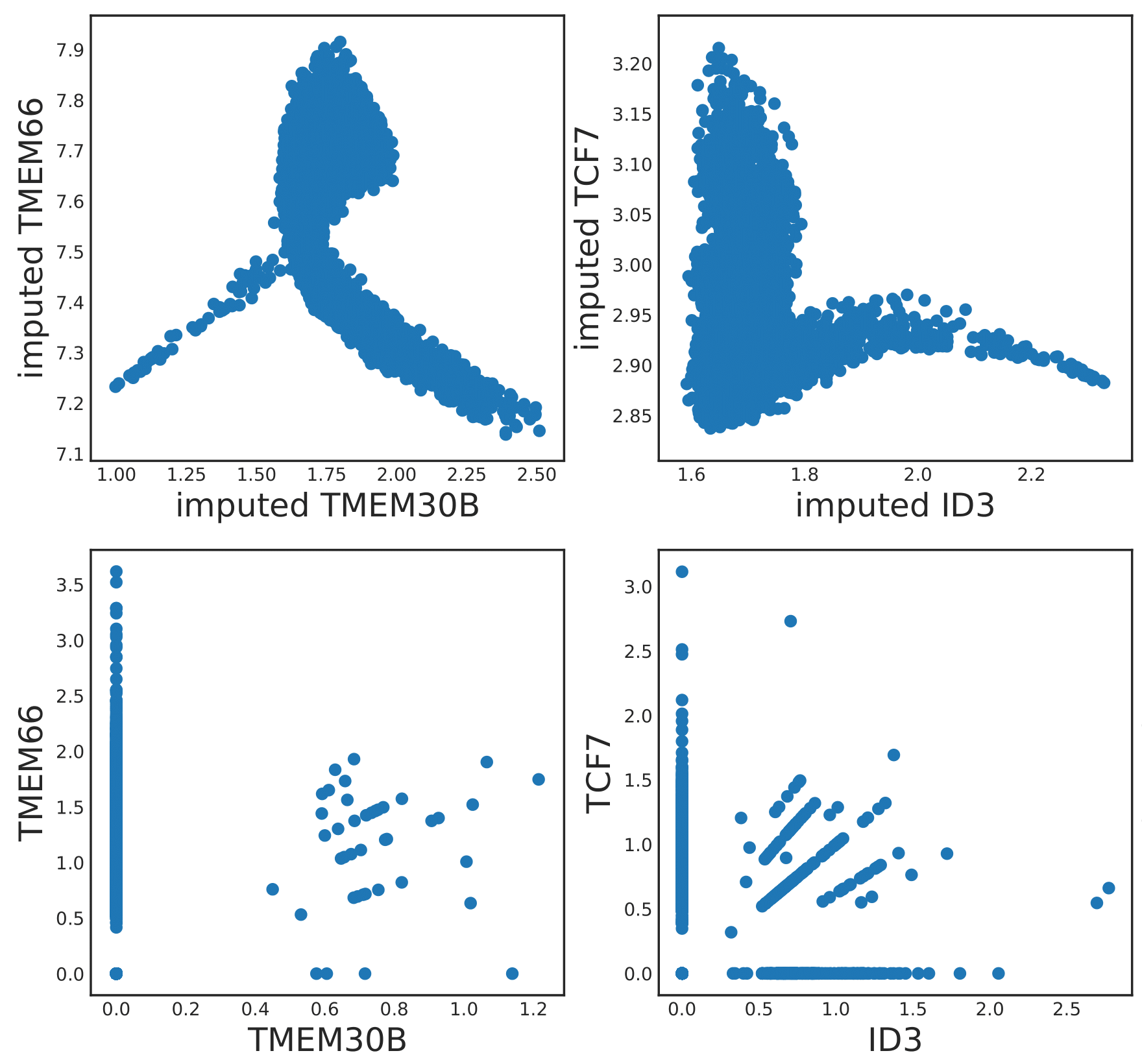}}
	\vspace{-0.15cm}
	\caption{(a) Anchored Causal Model with latent variables $Z_i$ and their corrupted observed counterparts $X_i$ (anchors). (b) Simulated Gaussian random variables before (top) and after dropout with probability 0.5 (bottom). (c) Imputed RNA-seq data (top) and raw data with dropout (bottom).} 
	\label{fig:model}
	\vspace{-0.3cm}
\end{figure}

Zhang et al.~\cite{zhang2017causal} considered the problem of learning a causal DAG model under measurement error as in Figure~\ref{fig:model}a, but restricted the measurement error to be independent from the latent variables. For many applications, including modeling dropout in single-cell RNA-seq data, this assumption is too restrictive. 
Halpern et al.~\cite{anchoredFA} considered more general anchored causal models than in Figure~\ref{fig:model}a, but only in the binary setting. 
Silva et al.~\cite{silva2006learning} considered a similar model for continuous distributions, but under the assumption that the dependence between latent and observed variables is linear, an assumption that is too restrictive for many applications. Inspired by topic modeling, the authors in~\cite{anandkumar2013learning} proposed a causal discovery method for DAGs with various levels of latent variables, but under the assumption that the latent variables are non-Gaussian and have sufficient outgoing edges for identifiability of the model.

The main contributions of this paper are as follows:
\vspace{-0.1cm}
\begin{itemize}
    \item We introduce \emph{anchored causal inference} to model causal relationships among latent variables from observations with measurement error in the Gaussian setting.
    \item We develop a provably consistent algorithm based on the method-of-moments to identify the causal structure (up to Markov equivalence) among the latent variables for a flexible class of measurement error models, including non-linear dependence on the latent variables. 
    \item In particular, we derive a consistent estimator for partial correlations and a consistent conditional independence test when the measurement noise is dropout.
    \item We present experimental results on both simulated and single-cell RNA-seq data, showing that our estimator, which takes into account measurement error, outperforms standard causal inference algorithms applied directly to the corrupted data.
\end{itemize}

\section{Preliminaries and Related Work}
\label{section:prelims}
Let $\mathcal{G} = ([p], E)$ be a directed acyclic graph (DAG) with nodes $[p]:=\{1, \dots , p\}$ and directed edges $E$. We associate a random variable $Z_i$ to each node $i\in[p]$. We denote the joint distribution of $Z=(Z_1, \dots , Z_p)^T$ by $\Prob$ and assume that $Z$ is generated by a \textit{linear Gaussian structural~equation~model}: 
\begin{equation}
    \label{eq:SEM}
    Z = B^TZ + \epsilon,
\end{equation}  
where $B$ is the weighted adjacency matrix of $\mathcal{G}$ and $\epsilon\sim\N_p(\pi, \Omega)$ with $\Omega =\textrm{diag}(\omega_1^2,\cdots,\omega_p^2)$. We consider the problem where only a noise-corrupted version of $Z$ is observed. We define $X_1,\cdots,X_p$ to be the \textit{observed} variables (\textit{anchors}) generated from the \textit{latent} variables $Z_1, \dots , Z_p$ by a noise process $X_i =F_i(Z_i)$, where $X_i$ has non-zero variance; see Figure~\ref{fig:model}a. We aim to learn the DAG $\mathcal{G}$ associated with the latent~variables~$Z$.




The majority of literature in causal inference assumes causal sufficiency, i.e.~no latent variables. A standard approach for causal structure discovery in this setting is to first infer the conditional independence (CI) relations among the observed variables and then use the CI relations to learn the DAG structure~\cite{spirtes2000causation}. However, since multiple DAGs can encode the same CI relations, $\mathcal{G}$ can only be identified up to its \emph{Markov equivalence class} (MEC). An MEC can be represented by a \emph{CPDAG}, a partially directed graph whose skeleton (underlying undirected graph) is the skeleton of $\mathcal{G}$ and an edge is directed if it has the same direction for all DAGs in the MEC~\cite{AMP97,VP90}. Various algorithms have been developed for learning a CPDAG under causal sufficiency~\cite{chickering2002optimal, solus2017consistency,spirtes2000causation}, most prominently the \emph{PC algorithm}~\cite{spirtes2000causation}, which treats causal inference as a constraint satisfaction problem with the constraints being CI relations. 
The PC algorithm is provably consistent, meaning that it outputs the correct MEC when the sample size $n\to\infty$, under the so-called \emph{faithfulness} assumption, which asserts that the CI relations entailed by $\Prob$ are exactly the relations implied by separation in the underlying DAG $\mathcal{G}$~\cite{spirtes2000causation}.

In the presence of latent variables, identifiability is further weakened  (only the so-called PAG is identifiable) and various algorithms have been developed for learning a PAG~\cite{RFCI,spirtes2001anytimeFCI,spirtes2000causation,zhang2008completenessFCI}. However, these algorithms cannot estimate causal relations among the latent variables, which is our problem of interest. Leung et al.~\cite{leung2016identifiability} study identifiability of directed Gaussian graphical models in the presence of a single latent variable. Zhang et al.~\cite{zhang2017causal}, Silva et al.~\cite{silva2006learning}, Halpern et al.~\cite{anchoredFA} and Anandkumar et al.~\cite{anandkumar2013learning} all consider the problem of  learning causal edges among latent variables from the observed variables, i.e.~models as in Figure~\ref{fig:model}a or generalizations thereof, but under assumptions that may not hold for our applications of interest, namely that the measurement error is independent of the latent variables~\cite{zhang2017causal}, that the observed variables are a linear function of the latent variables~\cite{silva2006learning}, that the observed variables are binary~\cite{anchoredFA}, or that each latent variable is non-Gaussian with sufficient outgoing edges to guarantee identifiability~\cite{anandkumar2013learning}.

\section{Anchored Causal Inference}
\label{section:model}

In the following, we first describe the assumptions of our \emph{Anchored Causal Model}, then motivate the model by the application to learning the underlying gene regulatory network from zero-inflated single-cell RNA-seq data, and finally provide an algorithm for anchored causal inference and prove its consistency under the model assumptions.




\begin{assumptions}[Anchored Causal Model]


\begin{assumpone} Given a DAG $\mathcal{G}=([p], E)$, the latent variables $Z=(Z_1, \dots Z_p)$ are generated by a linear Gaussian structural equation model (see (\ref{eq:SEM})) that is faithful to $\G$. 
\label{asmp:a1}
\end{assumpone}
\vspace{0.1cm}

\begin{assumptwo}
\label{asmp:a2} 
The observed random vector $X = (X_1,\dotsc,X_p)$ satisfies the CI relations $$X_i \indep \{X_1,\dotsc,X_p, Z_1,\dotsc,Z_p\}\setminus\{X_i,Z_i\} \;|\; Z_i\quad \textrm{ for all } i\in[p].$$ 
Furthermore, for all $i,j\in[p]$ there exists a finite-dimensional vector $\eta_i$ of monomials in $X_i$ and a finite-dimensional vector $\eta_{ij}$ of monomials in $X_i$ and $X_j$ such that their means 
can be mapped to the moments of the latent variables by known continuously differentiable functions $g_i$ and $g_{ij}$, i.e.~$\EX[Z_i] = g_{i}(\mathbb{E}[\eta_i])$ and $\EX[Z_iZ_j] = g_{ij}(\mathbb{E}[\eta_{ij}])$, and their covariance satisfies $\textrm{Cov}(\eta_i, \eta_{ij})<\infty$.
\end{assumptwo}
\end{assumptions}

While Assumption~\ref{asmp:a1} fixes the structural and functional relationship between the latent variables $Z$ by a linear Gaussian structural equation model, Assumption~\ref{asmp:a2} fixes the structural relationship between latent and observed variables by each $X_i$ having exactly one parent $Z_i$ for all $i\in[p]$, and ensures that the first- and second-order moments of $Z\sim\N_p(\mu, \Sigma)$ can be obtained from moments of $X$ without the restriction to a specific measurement error model. This allows for more general noise models than in~\cite{silva2006learning,zhang2017causal}. 


\begin{example}[Modeling single-cell RNA-seq data] 
\label{ex:rna_seq}
Let $Z_i$ represent the true latent RNA values (log-transformed) and $X_i$ the observed RNA values after dropout. The authors in~\cite{pierson2015zifa} considered a simple model of gene regulation represented by a linear Gaussian structural equation model among the latent variables $Z$ and modeled dropout for single-cell RNA-seq data by
\begin{equation*}
X_i = F_i(Z_i) = \begin{cases} 
Z_i & \textrm{ with probability} \quad q_i\\
0   & \textrm{ with probability} \quad 1-q_i
\end{cases}
\quad\textrm{for all}\quad i\in[p].
\end{equation*}
Assuming the dropout probabilities $q_i$ are known, this model satisfies Assumptions~\ref{asmp:a1} and \ref{asmp:a2} with
\begin{equation*}
\eta_i=X_i,\quad \eta_{ij}=X_iX_j, \quad  g_i(y) = \frac{y}{q_i}, 
\quad\textrm{and}\quad g_{ij}(y) = 
\begin{cases}
   \frac{y}{q_iq_j} & i\neq j\\
   \frac{y}{q_i} & i=j
\end{cases},
\end{equation*}
but does not satisfy the assumptions in~\cite{silva2006learning,zhang2017causal}. \qed
\end{example}

Algorithm~\ref{alg:general} describes our \emph{Anchored Causal Inference} procedure. The procedure works as follows: Given $n$ i.i.d.~samples of $X$ denoted by $\hat{X} = (\hat{X}^{(1)},\hat{X}^{(2)},\dotsc ,\hat{X}^{(n)})$, 
compute the required empirical moments $\mathbb{E}[\hat{\eta}_i]$ and $\mathbb{E}[\hat{\eta}_{ij}]$. Given a particular measurement error model defined by $g_i$ and $g_{ij}$, compute the first- and second-order moments of $Z$ to obtain its covariance matrix $\hat{\Sigma}$. If we can obtain the set of CI relations involving $Z$, we can use  causal structure discovery algorithms to learn $\G$ (up to its Markov equivalence class). Since $Z$ follows a Gaussian distribution, conditional independence corresponds to zero partial correlation. 
Let $i,j\in [p]$ and $K\subseteq [p]\setminus\{i,j\}$, then the sample partial correlations $\hat{\rho}_{ij\cdot K}$ can be computed recursively for increasing conditioning set sizes by
\begin{equation}
\label{eq:pcorr}
\hat{\rho}_{ij\cdot K } = 
\frac{\hat{\rho}_{ij\cdot K\setminus\{l\}}-\hat{\rho}_{il\cdot K\setminus\{l\}} \hat{\rho}_{jl\cdot K\setminus\{l\}}}{\sqrt{1-\hat{\rho}_{il\cdot K\setminus\{l\}\phantom{j}}^2}\sqrt{1-\hat{\rho}_{jl\cdot K\setminus\{l\}}^2}},
\end{equation}
where in the base case $\hat\rho_{ij\cdot \emptyset} = \hat\rho_{ij}$ are the correlations obtained from $\hat{\Sigma}$. The main difficulty lies in developing test statistics based on the estimated partial correlations $\hat\rho_{ij\cdot K}$ such that the inferred CI relations correspond as $n\to\infty$ to the set of CI relations implied by the underlying causal DAG $\G$. Such test statistics are developed in Corollaries~\ref{cor:transform} and~\ref{cor:variance}. The inferred CI relations can then be fed into a constraint-based causal discovery algorithm such as the PC algorithm~\cite{spirtes2000causation} to obtain the CPDAG~of~$\mathcal{G}$.

The first step in asserting consistency of Algorithm~\ref{alg:general} is the following lemma, which follows from the law of large numbers; the proof is provided in appendix~\ref{appendix:ProofOfLemmaMoments}.


\begin{algorithm}[!t]
\caption{Anchored Causal Inference}
\label{alg:general}
\textbf{Input:} $n$ samples $\hat{X} = (\hat{X}^{(1)},\hat{X}^{(2)},\dotsc ,\hat{X}^{(n)})$ of the random vector $X = F(Z)$.\\
\textbf{Output:} CPDAG representing the Markov equivalence class of the DAG $\mathcal{G}$ of the latent variables $Z$.\\
1. For each $i,j \in [p]$ compute the sample moment vectors 
$\mathbb{E}[\hat{\eta}_i]$ and $\mathbb{E}[\hat{\eta}_{ij}]$ from the samples $\hat{X}$. \\
2. Estimate the sample moments of $Z$ via $\hat{\mu}_i \triangleq g_i(\mathbb{E}[\hat{\eta}_i]),\; \hat{\mu}_{ij} \triangleq g_{ij}(\mathbb{E}[\hat{\eta}_{ij}])$.\\
3. Estimate the covariance matrix $\hat{\Sigma}$ of $Z$ by $(\hat{\Sigma})_{ij}= \hat{\mu}_{ij} - \hat{\mu}_i\hat{\mu}_j$ for all $i,j\in[p]$.\\
4. Estimate the partial correlations of $Z$ from $\hat{\Sigma}$ using \eqref{eq:pcorr}.\\
5. Calculate the test statistics defined in Corollaries~\ref{cor:transform} or~\ref{cor:variance} to infer the CI relations among the latent variables $Z$.\\
6. Use a consistent causal discovery algorithm (e.g.~the PC algorithm) based on the inferred CI relations.
\end{algorithm}

\begin{lemma}
\label{lemma:moments} 
Under assumptions~\ref{asmp:a1} and \ref{asmp:a2}, the estimator $\hat{\rho}_{ij\cdot K}$ in (\ref{eq:pcorr}) is asymptotically consistent.
\end{lemma}

Next, we design a consistent hypothesis test for obtaining CI relations based on the estimated partial correlations of $Z$ in (\ref{eq:pcorr}), similar in principle to Gaussian CI tests based on Fisher's z-transform used by many causal inference algorithms~\cite{chickering2002optimal,spirtes2000causation,wang2017permutation}. Under causal sufficiency, i.e., when $F_i$ is the identity function for all $i\in[p]$,
it can be shown~\cite{lehmann} that the estimated partial correlations in (\ref{eq:pcorr}) satisfy
\begin{equation}
\label{eq:GaussCorellationAsymptotic}
\sqrt{n}(\hat{\rho}_{ij\cdot K} - \rho_{ij\cdot K}) \xrightarrow{D}\N_1\Big(0, (1-\rho_{ij\cdot K}^2)^2\Big).
\end{equation}
Hence, applying the delta method to Fisher's z-transform $z_f(\rho) := (1/2)\log\Big((1+\rho)/(1-\rho)\Big)$ of the estimated partial correlations yields 
\begin{equation}
\label{eq:zTransformAsymptotic}
\sqrt{n}\Big(z_f(\hat{\rho}_{ij\cdot K}) - z_f(\rho_{ij\cdot K})\Big) \xrightarrow{D}\N_1(0,1).
\end{equation}
Hence Fisher's z-transform can be used in the test statistic $T:=\sqrt{n}\,z_f(\hat{\rho}_{i,j\cdot K})$ to test conditional independence by declaring $X_i\indep X_j | X_K$ at significance level $\alpha$ if and only if
\begin{equation}
\label{eq:generalCITest}
|T| \le \Phi^{-1}(1 - \frac{\alpha}{2}),
\end{equation}
where $\Phi^{-1}$ denotes the inverse CDF of $\mathcal{N}(0,1)$. The following theorem generalizes (\ref{eq:GaussCorellationAsymptotic}) to our anchored causal model class, where the partial correlations of $Z$ are estimated from the observed moments of $X$. 

\begin{theorem}
    \label{thrm:AsymptoticPartialCorrelations}
Let $\eta$ denote the vector of monomials of $X$ required to compute the first- and second-order moments of $Z$. Let $\nu$ denote the vector of first- and second-order moments of $\eta$. Then under assumptions~\ref{asmp:a1} and \ref{asmp:a2}, for any $i,j\in[p]$ and $K\subseteq[p]\setminus\{i,j\}$, the estimated partial correlation $\hat{\rho}_{ij\cdot K}$ in~(\ref{eq:pcorr}) satisfies
$$\sqrt{n}(\hat{\rho}_{ij\cdot K} - \rho_{ij\cdot K}) \xrightarrow{D}\N_1\big(0, \tau_{ij\cdot K}(\nu)\big)$$
    where 
    $\tau_{ij\cdot K}$ is a continuous function of $\nu$.
\end{theorem}

The proof of Theorem~\ref{thrm:AsymptoticPartialCorrelations} can be found in appendix~\ref{appendix:ProofOfTheoremAsymptoticPartialCorrelations}, 
where we provide a procedure for computing the function $\tau_{ij\cdot K}$ for any $i,j \in [p]$ and $K\subseteq[p]\setminus\{i,j\}$. The idea of the proof is as follows: First apply the Central Limit Theorem to the vector of sample moments $\mathbb{E}[\hat\eta]$. 
Under assumption~\ref{asmp:a2}, the correlations $\rho$ based on $\Sigma$ are continuously differentiable functions of $\nu$. Furthermore, for any $i,j\in [p]$ and $K\subseteq[p]\setminus\{i,j\}$, the partial correlation $\rho_{ij\cdot K}$ is defined recursively for increasing conditioning set sizes as a continuously differentiable function of $\rho$. Hence, one can iteratively apply the delta method starting from the statement of the Central Limit Theorem applied to $\EX[\hat\eta]$ to obtain the asymptotic distribution of $\hat{\rho}_{ij\cdot K}$.

In the following two corollaries to Theorem~\ref{thrm:AsymptoticPartialCorrelations}, we provide different test statistics for CI testing based on the estimated partial correlations of the latent vector $Z$. We start by generalizing Fisher's transform and its asymptotic distribution given in (\ref{eq:zTransformAsymptotic}).

\begin{corollary}
\label{cor:transform}
If the asymptotic variance $\tau_{ij\cdot K}(\nu)$ can be written purely as a function of $\rho_{ij\cdot K}$, i.e., there exists $\tilde\tau_{ij\cdot K}$ such that $\tau_{ij\cdot K}(\nu) = \tilde{\tau}_{ij\cdot K }(\rho_{ij\cdot K})$, and  there exists a variance stabilizing transformation $z_{ij\cdot K}$ such that
\begin{equation}
    \label{int_1}
z_{ij\cdot K}(\rho) = \int \frac1{\sqrt{\tilde\tau_{ij\cdot K}(\rho)}}d\rho + C
\end{equation}
with $C$ chosen such that $z_{ij\cdot K}(0) = 0$, then under~\ref{asmp:a1} and \ref{asmp:a2}
$$\sqrt{n} \Big(z_{ij\cdot K}(\hat{\rho}_{ij\cdot K }) - z_{ij\cdot K}(\rho_{ij\cdot K})\Big) \xrightarrow{D} \N_1(0,1).$$
\end{corollary}
The proof of Corollary~\ref{cor:transform} follows by applying the delta method to  Theorem~\ref{thrm:AsymptoticPartialCorrelations} (appendix~\ref{appendix:PoorfOfCorollaryTransform}).
Whether the conditions  of  Corollary~\ref{cor:transform} are satisfied, depends on the measurement error model $F$. We show in appendix~\ref{appendix:DerivationOfDropoutStabilizing} that the conditions of Corollary~\ref{cor:transform} hold for the dropout model for $K =\emptyset$ and $\mu = 0$, and derive the corresponding variance stabilizing transformation.
Note that it is sufficient if we can compute the integral in (\ref{int_1}) numerically; a closed-form solution is not required. Corollary~\ref{cor:transform} implies that the test statistic 
$T=\sqrt{n}\,z(\hat{\rho}_{i,j\cdot K})$ in (\ref{eq:generalCITest}) can be used to consistently estimate the CI relations among the latent variables $Z$.
When the assumptions of Corollary~\ref{cor:transform} are not met, then we can obtain a different test statistic that is asymptotically normal as shown in the following result.

\begin{corollary}
\label{cor:variance} Define $\zeta_{ij\cdot K}(\hat{\rho}, \hat\nu) := {\hat{\rho}}/{\sqrt{\tau_{ij\cdot K}({\hat\nu})}}$. Then under~\ref{asmp:a1} and \ref{asmp:a2}
$$\sqrt{n}\bigg(\zeta_{ij\cdot K}(\hat{\rho}_{ij\cdot K}, \hat\nu) - \zeta_{ij\cdot K}(\rho_{ij\cdot K},\hat\nu)\bigg) \xrightarrow{D} \N_1(0, 1).$$
\end{corollary}

Corollary~\ref{cor:variance} follows from Theorem~\ref{thrm:AsymptoticPartialCorrelations}: since $\tau$ is continuous in $\nu$, $\hat\nu$ converges to $\nu$ by the law of large numbers and implies that $\tau_{ij\cdot K}(\hat\rho_{ij\cdot K},\hat\nu)\xrightarrow{a.s.} \tau_{ij\cdot K}(\hat\rho_{ij\cdot K}, \nu)$ as $n\rightarrow\infty$ (appendix~\ref{appendix:ProofOfCorollaryVariance}). Hence the test statistic $T =\sqrt{n}\,\zeta_{ij\cdot K}(\hat{\rho}_{ij\cdot K}, \hat\nu)$ can be used in (\ref{eq:generalCITest}) to obtain the CI relations among the latent variables $Z$.

With respect to finite-sample considerations, note that $\zeta_{ij\cdot K}$ in Corollary~\ref{cor:variance} is a function of $\hat\nu$, and thus its convergence to its asymptotic distribution requires the convergence of $\hat\nu$. Hence, we expect the convergence in distribution of Corollary~\ref{cor:transform} to be faster than that of Corollary~\ref{cor:variance} and as a result the test statistic in Corollary~\ref{cor:transform} to perform better in the finite-sample regime.

We end this section with the main result of this paper, namely the consistency of Algorithm~\ref{alg:general}.
\begin{theorem}
\label{thrm:AlgorithmConsistency}
Under assumptions~\ref{asmp:a1} and \ref{asmp:a2},  Algorithm \ref{alg:general} is consistent, i.e., 
as $n\rightarrow\infty$ it returns a CPDAG that represents the Markov equivalence class of the true DAG $\G$.
\end{theorem}
The proof can be found in appendix~\ref{appendix:ProofOfAlgorithmConsistency}. In particular, we show that under the faithfulness assumption, the set of CI relations inferred from $\hat \Sigma$ in  Algorithm~\ref{alg:general} 
converges to the set of CI statements implied by the underlying DAG $\mathcal{G}$. Hence using any consistent causal structure discovery algorithm on these CI relations results in the correct Markov equivalence class.

\section{Implementation}
\label{section:implementation}
Next, we discuss an important aspect of implementation and show how the results in Section~\ref{section:model} can be applied to the dropout model in Example~\ref{ex:rna_seq}.

In the finite-sample setting, the estimated covariance matrix $\hat{\Sigma}$ of the latent variables is not guaranteed to be positive semidefinite. 
In this case, shrinkage towards a positive definite matrix can be used as a form of regularization. 
When $n<p$, a standard approach is to use Ledoit-Wolf shrinkage towards the identity matrix~\cite{shrinkage}. When $n>p$, the sample covariance matrix $\hat{S}$ based on the samples $\hat{X}$ is positive definite with probability 1 and hence $\hat{\Sigma}$ can be shrunk towards $\hat{S}$ by
\begin{equation}
\label{eq:shrinkage}
    \hat{\Lambda} = (1-\alpha^*)\hat{\Sigma} + \alpha^*\hat{S} \quad \text{where} \quad \alpha^* = \argmin_{\alpha \in [0,1],\hat{\Lambda}  \succeq 0} \alpha .
\end{equation}
The form of shrinkage that provides better results depends on whether $\hat{S}$ or the identity matrix are better approximations of the true underlying covariance matrix $\Sigma$. In our experiments in Section~\ref{sec:exp}, we applied shrinkage towards $\hat{S}$ as in~\eqref{eq:shrinkage}.
Both types of shrinkage result in consistent estimates: consistency of Ledoit-Wolf shrinkage is proven in~\cite{shrinkage} and the consistency of shrinkage towards the sample covariance matrix in (\ref{eq:shrinkage}) follows from Theorem~\ref{lemma:moments}, since $\hat{\Sigma} \rightarrow \Sigma$ as $n\rightarrow \infty$ implies that $\hat{\Sigma}$ becomes positive semidefinite with large enough sample size, and therefore, $\alpha\rightarrow 0$ as $n\to\infty$, which shows that shrinkage reduces to the consistent case without shrinkage.


\subsection{Application: The Dropout Model}


Under the dropout model in Example~\ref{ex:rna_seq}, the assumptions of Corollary~\ref{cor:transform} are not satisfied, since in general $\tau_{ij\cdot K}$ cannot be expressed as a function of $\rho_{ij\cdot K}$ only. This is shown in appendix~\ref{appendix:DropoutStabilizingTransformConditionPlots} by plotting $\tau$ as a function of $\nu$ for fixed $\rho$.  In the special case when $\mu=0$, the conditions are satisfied for all $i,j\in[p]$ when $K=\emptyset$, i.e., a variance stabilizing transformation $z_{ij} = z_{ij\cdot \emptyset}$ can be found for the correlations $\rho_{ij} = \rho_{ij\cdot \emptyset}$. This \textit{dropout stabilizing transform} is provided in appendix~\ref{appendix:DerivationOfDropoutStabilizing}. This transform and the resulting CI test can be used as a heuristic also when $K\neq \emptyset$ or $\mu\neq 0$ and we analyze its performance in Section~\ref{sec:exp}.

In appendix~\ref{appendix:DerivationOfNormalizing}, we provide a recursive formula for computing $\zeta_{ij\cdot K}$ from Corollary~\ref{cor:variance} for the dropout model, which we refer to as the \textit{dropout normalizing transform}. 
Computing it requires determining the asymptotic variance $\tau_{ij\cdot K}$ of $\rho_{ij\cdot K}$. Also note that applying shrinkage will in general change the asymptotic variance of the partial correlations. We show how to correct $\tau_{ij\cdot K}$ as a function of the shrinkage coefficient $\alpha$ in appendix~\ref{appendix:DerivationOfNormalizingShrinkage}. 
This adjustment is applied in all of our experiments in Section~\ref{sec:exp}. 




In Section~\ref{section:RealData}, we apply our estimation procedure based on the dropout model to single-cell RNA-seq data in order to infer the structure of the underlying gene regulatory network. For the theoretical analysis in Section~\ref{section:model} we assumed that the dropout probabilities $q_i$ are known. However, in general these parameters need to be estimated and it was proposed in~\cite{emma2015a} to model $q_i$ by $q_i =1- \exp^{\lambda \mu_i^2}$ for $i \in [p]$,
where $\lambda$ depends on the single-cell RNA-seq assay.
Using this model for the dropout probabilities we can jointly estimate the parameters $\mu$ and $q$ as follows. This model implies
\begin{equation*}
    \EX[\eta_i] = \EX[X_i] =  \EX[(1 - \exp^{\lambda \mu_i^2})Z_i] = (1 - \exp^{\lambda \mu_i^2})\mu_i.
\end{equation*}
Since $\mu_i$ corresponds to the gene expression count averages, we can assume that $\hat\mu_i\geq 0$. Under this assumption, the equation
$\EX[\hat\eta_i] = (1 - \exp^{\lambda \hat\mu_i^2})\hat\mu_i$ has a unique solution 
for $\hat\mu_i$. 
With respect to the parameter $\lambda$, for some single-cell RNA-seq assays it is possible to obtain an estimate for $\lambda$ by including molecules with known expression as controls. However, since this estimate is often unreliable~\cite{grun2015design} and not always available, we selected $\lambda$ so as to minimize the amount of shrinkage required to obtain a positive semidefinite matrix.





\section{Experiments}
\label{sec:exp}
    In this section, we analyze the performance of Algorithm~\ref{alg:general} based on the dropout model both on simulated data and on single-cell RNA-seq data.
    
\subsection{Simulations}

\begin{figure}[t!]
	\centering
	\subfigure[ROC Skeleton]{\includegraphics[width=.245\linewidth]{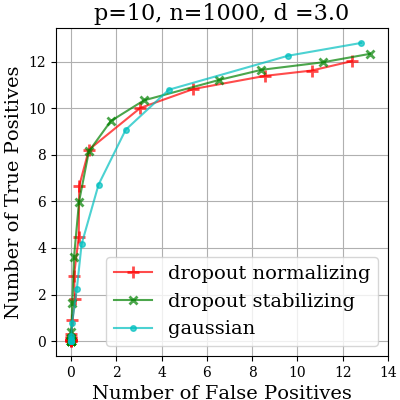}}
	\subfigure[ROC Skeleton]{\includegraphics[width=.245\linewidth]{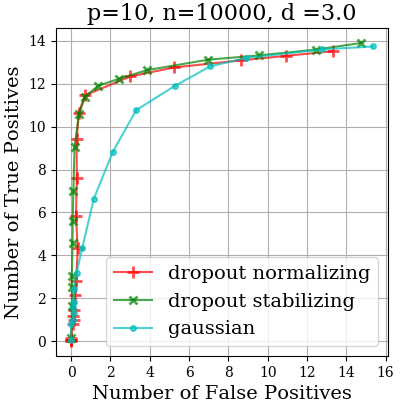}}
	\subfigure[SHD Skeleton]{\includegraphics[width=.245\linewidth]{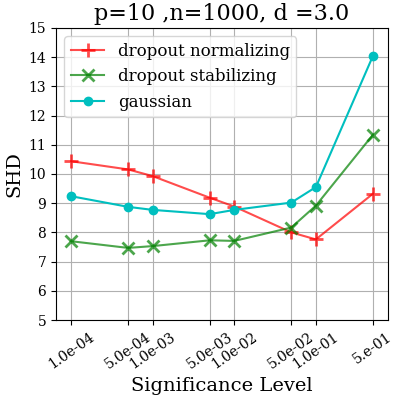}} 
	\subfigure[SHD Skeleton]{\includegraphics[width=.245\linewidth]{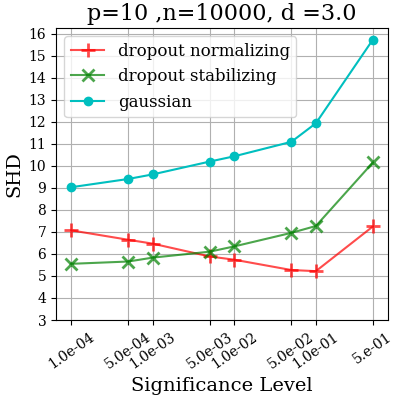}} \\
	
	\subfigure[ROC Skeleton]{\includegraphics[width=.245\linewidth]{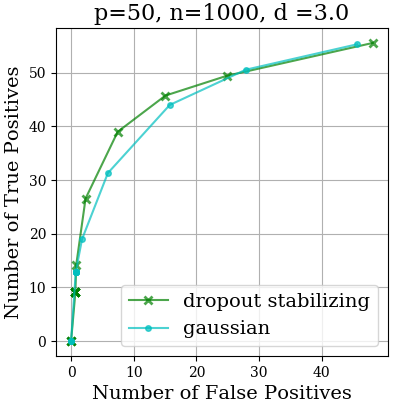}}
	\subfigure[ROC Skeleton]{\includegraphics[width=.245\linewidth]{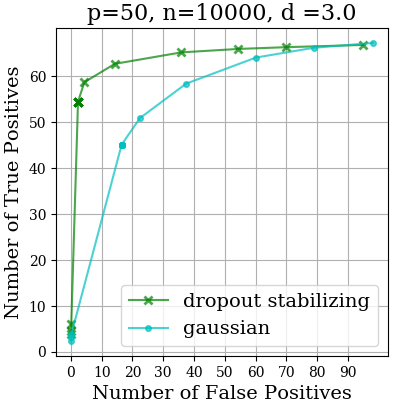}}
	\subfigure[SHD Skeleton]{\includegraphics[width=.245\linewidth]{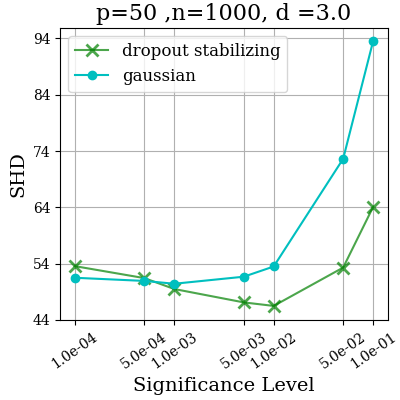}} 
	\subfigure[SHD Skeleton]{\includegraphics[width=.245\linewidth]{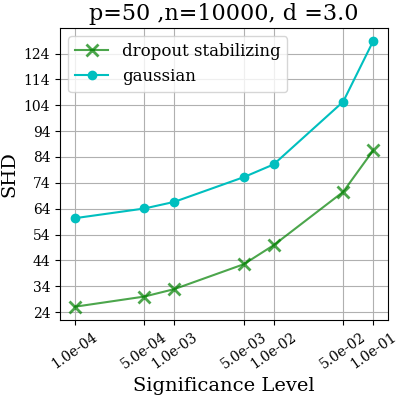}} 
	\caption{Performance of the dropout stabilizing transform, dropout normalizing transform and Gaussian CI test in simulations. (a)-(d) ROC and SHD curves with $p=10$ and (e-h) with $p=50$ for evaluating the accuracy of estimating the skeleton of the true DAG.}
	\vspace{-0.3cm}
	\label{fig:simulations}
\end{figure}

\label{section:simulations}
The data was generated from the dropout model described in Example~\ref{ex:rna_seq}.
The structure of the matrix $B$ in the linear Gaussian structural equation model~\eqref{eq:SEM} was generated by an Erd\"os-Renyi model with expected degree $d$ for $d\in\{2, 3, 5\}$ and number of nodes $p\in\{10, 30, 50\}$. The weights of the matrix $B$ were uniformly drawn from $[-1,-0.25] \cup [0.25,1]$ to be bounded away from $0$. 
The mean parameters $\mu_i$ were uniformly drawn from $[0,3]$ and the probabilities $q_i$ from $[0,0.8]$.
These ranges were chosen to match the expected ranges in the gene expression data analyzed in Section~\ref{section:RealData}. We generated $n$ observations of $X$ from this generating model, for $n\in\{1000, 2000, 10000, 50000\}$. 

Figure~\ref{fig:simulations} shows the ROC curves and the Structural Hamming Distance (SHD) plots evaluating the skeleton of the CPDAG output by Algorithm~\ref{alg:general} for $p\in\{10,50\}$, $n\in\{1000,10000\}$ and $d=3$. Each point in the plots is an average over $96$ simulations (divisible by the number of cores). The CI relations were obtained using the dropout stabilizing transform (green curve), the dropout normalizing transform (red curve), and the Gaussian CI test applied directly to the perturbed observations (blue curve) and the PC algorithm was used to estimate the CPDAG. 
In appendix~\ref{appendix:MorePlots}, we provide additional figures for the setting where $p=30$, $n\in\{2000, 50000\}$ and $d\in\{2,5\}$, as well as the ROC curves and SHD plots evaluating the inferred CPDAG.

The simulation results show that the dropout stabilizing transform outperforms, or performs at least as well as the Gaussian CI test in all settings we tested even though it was derived for $K=\emptyset$ and $\mu=0$. The performance of both dropout transforms  improves over the Gaussian CI test with increasing sample size. A large sample size is especially important for the dropout normalizing transform because it relies on the estimation of more parameters.
Since the dropout stabilizing transform is   preferable to  the dropout normalizing transform from a computational point of view while performing similarly well, we concentrate on this transform in Figure~\ref{fig:simulations}(e-h). 

The simulation results show that our estimators significantly outperform the naive Gaussian CI test applied directly to the corrupted data. This has important implications for the development of new single-cell RNA-seq technologies indicating that an increased sample size is preferrable to minimizing dropout. Current single-cell technologies have been heading exactly in this direction, trading off increased sample sizes (with studies containing up to a million samples) for an increased dropout rate~\cite{single_cell_millions}. The simulation results suggest that our estimators are well-suited for such data. 

\subsection{Single-cell RNA-seq Data}

\begin{figure}[t!]
	\centering
	\subfigure[Perturb-seq]{\includegraphics[width=.245\textwidth]{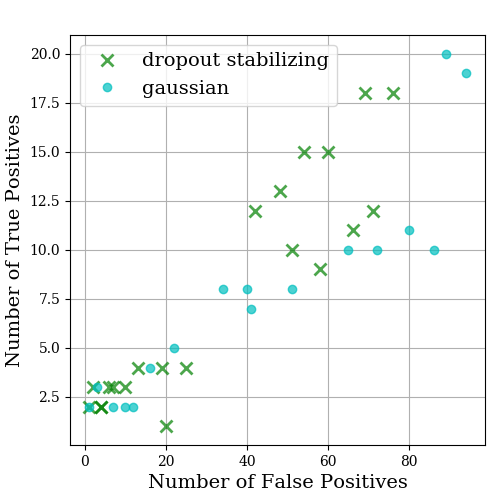}} 
	\subfigure[Perturb-seq gene regulatory network]{\includegraphics[width=.45\textwidth]{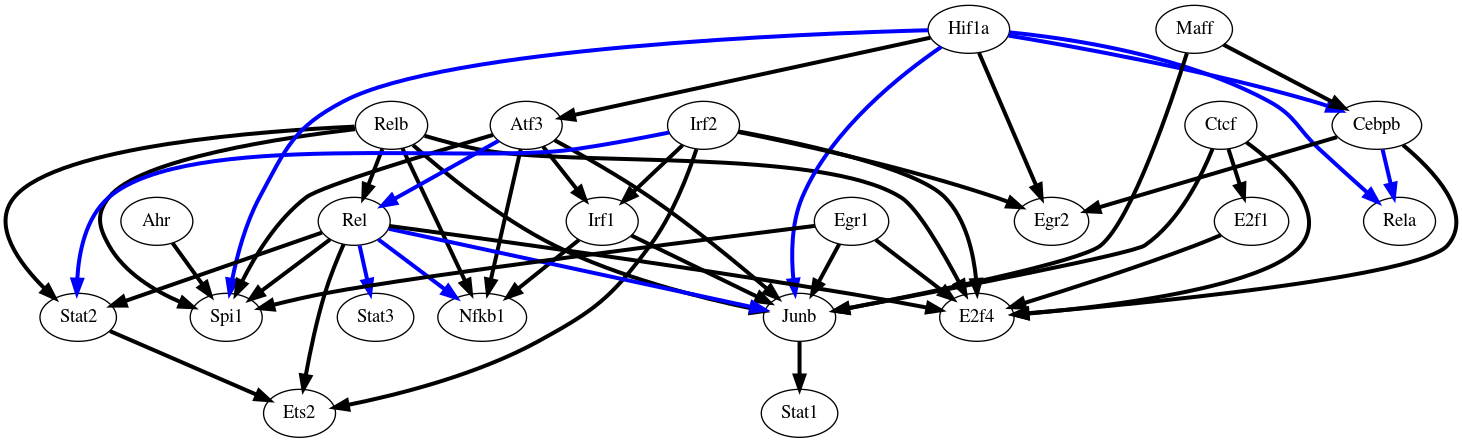}}
	\subfigure[Pancreas]{\includegraphics[width=.245\textwidth]{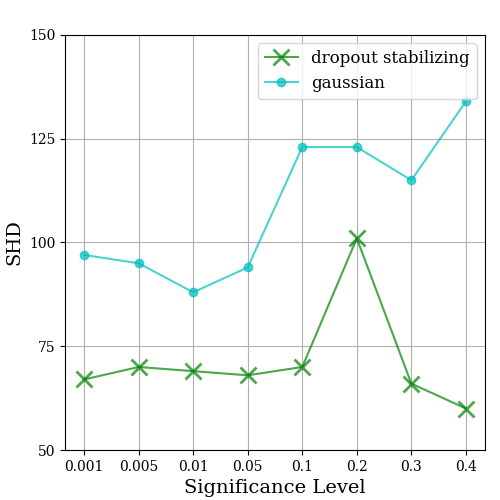}}
	\caption{(a) ROC curve for predicting causal effects of interventions in Perturb-seq data. (b) Gene regulatory network estimated from Perturb-seq data (blue edges indicate previously known interactions~\cite{dixit2016perturb} that were also detected by our method). (c) SHD between CPDAGs estimated on data sets collected with low versus high dropout rate in the pancreas.}
	\label{fig:real_data}
	\vspace{-0.4cm}
\end{figure}

\label{section:RealData}
\textbf{Perturb-seq.} We tested our method on gene expression data collected via single-cell Perturb-seq by Dixit et al.~\cite{dixit2016perturb} from bone marrow-derived dendritic cells (BMDCs). As in most single-cell studies, the gene expression observations are affected by dropout.  The data consists of 933 observational samples (after standard pre-processing), which we used for learning the gene regulatory network. The Perturb-seq data set also contains interventional samples, which we used to evaluate the estimated CPDAG and construct an ROC curve. As in~\cite{dixit2016perturb}, we focussed our analysis on 24 genes, which are important transcription factors known to regulate each other as well as a variety of other genes~\cite{garber2012high}. We used the dropout stabilizing transform to obtain the CI relations among the latent variables and compared the resulting CPDAG to the graph obtained using the standard Gaussian CI test applied directly to the observed corrupted data. In both settings we used the PC algorithm to infer the CPDAG from  the CI relations. 
Figure~\ref{fig:real_data}a shows the resulting ROC curve, which quantifies for varying tuning parameters the accuracy of each of the learned CPDAGs
in predicting the effect of each of the eight interventions as also described in~\cite{wang2017permutation}. Our algorithm with the dropout stabilizing transform outperforms the Gaussian CI test. The inferred gene regulatory network is shown in Figure~\ref{fig:real_data}b.

\textbf{Pancreas - Type II Diabetes.} We also tested our method on two gene expression data sets collected from human pancreatic cells~\cite{baron2016single, segerstolpe2016single} via different single-cell assays, one with low dropout rate (Smart-seq2) and the other with high dropout rate (inDrop). The Smart-seq2 data set consists of 3514 cells and the inDrop data set of 8569 cells. We focused our analysis on a gene regulatory network of 20 genes, which is known to be involved in Type II Diabetes~\cite{sharma2018controllability}. Since no interventional data is available for this application, we evaluated our estimator based on how consistent the estimated CPDAG is across the two data sets. Figure~\ref{fig:real_data}c shows the SHD between the  CPDAGs estimated from the data set with low versus high dropout using the dropout stabilizing transform as compared to the Gaussian CI test applied directly on the observed data. The inferred gene regulatory network  is provided in appendix~\ref{appendix:MorePlots}.
 Since the SHD is lower for the dropout stabilizing transform than the Gaussian CI test, the CPDAG estimates produced by our method are more consistent across different dropout levels, thereby suggesting that our method is more robust to dropout.

%

\section{Discussion}
In this paper, we proposed a procedure for learning causal relationships in the presence of measurement error. For this purpose, we considered the  \emph{anchored causal model}, where each corrupted observed variable is generated from a latent uncorrupted variable and the aim is to learn the causal relationships among the latent variables using the observed variables as \emph{anchors}. We introduced an algorithm that learns the Markov equivalence class of the causal DAG among the latent variables based on the empirical moments of the observed variables and proved its consistency.  One of the main motivations for developing this algorithm was to address the problem of dropout in single-cell RNA-seq experiments. We showed how to apply our algorithm for learning the underlying gene regulatory network under a standard dropout model and analyzed its performance on synthetic data and on single-cell RNA-seq data, thereby showing that taking into account dropout  allows identifying the causal relationships between genes in a more accurate and robust manner.

\bibliographystyle{plainnat}  
\bibliography{references}  

\newpage

\appendix
\begin{appendices}

\section{Proof of Lemma \ref{lemma:moments}}
    \label{appendix:ProofOfLemmaMoments}
    \begin{proof}
    By the strong law of large numbers, for all $i,j\in [p]$, 
    $$\EX[\hat\eta_{i}]\xrightarrow{a.s.}\EX[\eta_i] \quad\textrm{and}\quad \EX[\hat\eta_{ij}]\xrightarrow{a.s.}\EX[\eta_{ij}] \quad\textrm{as}\quad n\rightarrow\infty.$$ 
    The functions $g$ are continuous since they are continuously differentiable by~\ref{asmp:a2}, therefore, for all $i\in[p]$,
    $$\hat{\mu}_i = g_i(\EX[\hat{\eta}_i]) \xrightarrow{a.s.}g_i(\EX[\eta_i]) = \mu_i  \quad\textrm{as}\quad n\rightarrow\infty$$
    and similarly, for all $i,j\in[p],$
    $$\hat{\mu}_{ij} = g_{ij}(\EX[\hat{\eta}_{ij}]) \xrightarrow{a.s.} g_{ij}(\EX[\eta_{ij}]) = \mu_{ij}  \quad\textrm{as}\quad n\rightarrow\infty.$$
    Therefore,
    $$(\hat\Sigma)_{ij} := \hat\mu_{ij} - \hat\mu_i\hat\mu_j\xrightarrow{a.s.}\mu_{ij} - \mu_i\mu_j = (\Sigma)_{ij} \quad\textrm{as}\quad n\rightarrow\infty$$
    and the correlations
    $$\hat\rho_{ij} = \frac{(\hat\Sigma)_{ij}}{\sqrt{(\hat\Sigma)_{ii}(\hat\Sigma)_{jj}}} \xrightarrow{a.s.} \frac{(\Sigma)_{ij}}{\sqrt{(\Sigma)_{ii}(\Sigma)_{jj}}} = \rho_{ij} \quad\textrm{as}\quad{n\rightarrow\infty}.$$
    Recursively applying a similar argument to equation~\eqref{eq:pcorr} proves that $\hat\rho_{ij\cdot K}$ is consistent.
    \end{proof}

\section{Proof of Theorem \ref{thrm:AsymptoticPartialCorrelations}}
\label{appendix:ProofOfTheoremAsymptoticPartialCorrelations}
We start by defining the vectors of all correlations estimated from Algorithm~\ref{alg:general} and all true correlations of $Z$ as
\begin{equation}
    \label{eq:AllCorrs}
    \hat\rho := \begin{pmatrix}
    \hat\rho_{12}\\
    \hat\rho_{13}\\
    \cdot\\
    \cdot\\
    \cdot\\
    \hat\rho_{(p-1)p}\\
    \end{pmatrix}
    \quad\textrm{and}\quad
    \rho := \begin{pmatrix}
    \rho_{12}\\
    \rho_{13}\\
    \cdot\\
    \cdot\\
    \cdot\\
    \rho_{(p-1)p}\\
    \end{pmatrix},
\end{equation}
respectively.
We use $\eta$ to denote the vector obtained from concatenating all monomials in $X_i$ and $X_j$ that appear in $\eta_i$ and $\eta_{ij}$ for $i,j\in[p]$ in Assumption~\ref{asmp:a2}.
That is,
\begin{equation*}
    \eta := 
    \begin{pmatrix}
    \eta_{1}^T &
    \eta_{2}^T &
    \dots &
    \eta_{p}^T &
    \dots &
    \eta_{11}^T &
    \eta_{12}^T &
    \eta_{pp}^T
    \end{pmatrix}.
\end{equation*}
We let
\begin{equation*}
    \hat\eta := 
    \begin{pmatrix}
    \hat\eta_{1}^T &
    \hat\eta_{2}^T &
    \dots &
    \hat\eta_{p}^T &
    \dots &
    \hat\eta_{11}^T &
    \hat\eta_{12}^T &
    \hat\eta_{pp}^T
    \end{pmatrix}.
\end{equation*}
be the analogous concatenated vector of sample monomials in $\hat{X}_i$ and $\hat{X}_j$ calculated from the data 
$\hat{X}= (\hat{X}^{(1)},\hat{X}^{(2)},\dotsc,\hat{X}^{(n)})$.

The following lemma is concerned with the asymptotic distribution of the correlation vector $\hat\rho$.

\begin{lemma}
\label{lemma:CorrDelta}
Under Assumptions~\ref{asmp:a1} and~\ref{asmp:a2}, 
$$\sqrt{n}(\hat{\rho} - \rho) \xrightarrow{D} \mathcal{N}_{|\rho|}\Big(0, A(\nu)\Big),$$
where $\nu$ is the vector of all first and second order moments of $\eta$ and $A$ is a continuous function of $\nu$.
\end{lemma}

\begin{proof}

Assumption~\ref{asmp:a2} asserts that the covariance of $\eta$ is finite. Hence, we can apply the Central Limit Theorem to obtain
\begin{equation}
    \label{eq:EtaConvergence}
    \sqrt{n}(\EX[\hat\eta] - \EX[\eta]) \xrightarrow{D} \N_{|\eta|}(0, A_\eta(\nu)),
\end{equation}
where $A_\eta$ is the covariance matrix of $\eta$. The elements of the covariance matrix $A_\eta$ can be written as a continuous function of the first-and second-order moments of $\eta$, i.e., they can be written as a continuous function of $\nu$. 

Assumptions~\ref{asmp:a1} and~\ref{asmp:a2} imply that we can write for all $i,j\in [p]$,
\begin{equation}
\label{eq:EtaToRho}
\rho_{ij} = \frac{g_{ij}(\EX[\eta_{ij}])-g_i(\EX[\eta_i])g_j(\EX[\eta_j])}{\sqrt{g_{ii}(\EX[\eta_{ii}])-g_i(\EX[\eta_i])^2}\sqrt{g_{jj}(\EX[\eta_{jj}])-g_j(\EX[\eta_j])^2}}.
\end{equation}
We compute the sample correlation $\hat\rho_{ij}$ in our algorithm as
\begin{equation*}
\hat\rho_{ij} = \frac{g_{ij}(\EX[\hat\eta_{ij}])-g_i(\EX[\hat\eta_i])g_j(\EX[\hat\eta_j])}{\sqrt{g_{ii}(\EX[\hat\eta_{ii}])-g_i(\EX[\hat\eta_i])^2}\sqrt{g_{jj}(\EX[\hat\eta_{jj}])-g_j(\EX[\hat\eta_j])^2}}.
\end{equation*}
Based on equation~\eqref{eq:EtaToRho} we can define a function $w:\mathbb{R}^{|\eta|}\rightarrow\mathbb{R}^{|\rho|}$ such that $w(\EX[\eta]) = \rho$ and $w(\EX[\hat\eta]) = \hat\rho$. Applying the delta method to equation~\eqref{eq:EtaConvergence} with the function $w$, we get 
\begin{equation*}
    \sqrt{n}(\hat\rho - \rho ) \xrightarrow{D} \N_{|\rho|}\Big(0, A_\rho(\nu)\Big),
\end{equation*}
where $A_\rho(\nu) = \nabla w(\EX[\eta])^T A_\eta(\nu)\nabla w(\EX[\eta])$, since the elements of the mean vector $\EX[\eta]$ are elements of $\nu$.
Notice that under the assumption that the variance of $X_i = F_i(Z_i)$ is non-zero from Section~\ref{section:prelims}, the denominator in~\eqref{eq:EtaToRho} is non-zero, and therefore $\rho$ is continuously differentiable in $g(\EX[\eta])$, which is continuously differentiable in $\EX[\eta]$ by Assumption~\ref{asmp:a2}. Hence, $\nabla w(\EX[\eta])$ is continuous in $\EX[\eta]$, and therefore continuous in $\nu$. Since $A_{\rho}(\nu)$ is a matrix product of functions continuous in $\nu$, it is also continuous in $\nu$.
\end{proof}

\begin{lemma}
\label{lemma:CorrToPcorrDelta}
If
$$\sqrt{n}(\hat{\rho} - \rho) \xrightarrow{D} \N_{|\rho|}\Big(0, A_{\rho}(\nu)\Big)$$
where $\nu$ is the vector of all first- and second-order moments of $\eta$, and $A_{\rho}(\nu)$ is continuous in $\nu$, then under Assumptions~\ref{asmp:a1} and~\ref{asmp:a2}, 
for any $i,j \in [p]$ and $K\subseteq [p] \setminus \{i,j\}$,
$$\sqrt{n}(\hat{\rho}_{ij\cdot K} - \rho_{ij\cdot K}) \xrightarrow{D} \mathcal{N}_1(0, \tau_{ij\cdot K}(\nu)),$$ 
for some $\tau_{ij\cdot K}$ that is continuous in $\nu$, where $\hat\rho_{ij\cdot K}$ are the partial correlations estimated by Algorithm~\ref{alg:general}.
\end{lemma}

\begin{proof}
We take any arbitrary but fixed $i,j \in [p]$ and subset $K \in [p]$, and we prove the lemma for $\hat{\rho}_{ij\cdot K}.$ Let $k:= |K| +2$, where $|K|$ is the size of the conditioning set $K$. We begin by relabeling the variables of interest for clarity. We relabel $i$ to $1$, $j$ to $2$ and the elements of $K$ to $S =\{3,\cdots, k\}$.
 Furthermore, we define the sets
 $$S_m:= \begin{cases}
    \{m,m+1,\cdots,k\} & 3\le m \le k\\
     \emptyset & m = k+1
 \end{cases}.$$
 Note that $S_3 = S$, and thereby, the partial correlation of interest is ${\rho}_{12\cdot S} = \rho_{12\cdot S_3}$. Now we define for $m\in\{3,\cdots, k+1\}$, the vectors
$$
\hat{\rho}_{m} := \begin{pmatrix}
\hat{\rho}_{1,2 \cdot  S_{m}}\\
\hat{\rho}_{1,3\cdot   S_{m}}\\
\hat{\rho}_{2,3\cdot   S_{m}}\\
\hat{\rho}_{1,4\cdot   S_{m}}\\
\hat{\rho}_{2,4\cdot   S_{m}}\\
\dotsc\\ 
\hat{\rho}_{1, m-1\cdot   S_{m}}\\
\hat{\rho}_{2, m-1\cdot   S_{m}}\\
\dotsc\\
\hat{\rho}_{{m-2},m-1\cdot S_{m}}\\
\end{pmatrix}
\quad\textrm{and}\quad
\rho_{m} := \begin{pmatrix}
\rho_{1, 2 \cdot  S_{m}}\\
\rho_{1, 3\cdot   S_{m}}\\
\rho_{2, 3\cdot   S_{m}}\\
\rho_{1, 4\cdot   S_{m}}\\
\rho_{2, 4\cdot   S_{m}}\\
\dotsc\\ 
\rho_{1, m-1\cdot   S_{m}}\\
\rho_{2, m-1\cdot   S_{m}}\\
\dotsc\\
\rho_{{m-2},m-1\cdot S_{m}}\\
\end{pmatrix}.
$$
It follows from the definition of $S_{k+1}$ that $\hat{\rho}_{k+1} = \hat\rho$ and $\rho_{k+1} = \rho$. In order to prove the lemma, we proceed by induction on $m$ starting with the base case of $m= k+1$ and show that for all $m$ such that $3\le m \le k$,

\begin{equation}
\label{eq:invariant}
\sqrt{n}(\hat{\rho}_{m}- \rho_{m}) \xrightarrow{D} \N_{|\rho_{m}|}\Big(0, A_{m}(\nu)\Big)
\end{equation}

for some $A_{m}$ that is continuous in $\nu$. Note that the base case is given by the hypothesis in the lemma.
Moreover, the statement of the lemma is that the above holds for $m = 3$, and therefore, completing the inductive step proves the lemma.

To complete the inductive step, assume that for $m$ such that $3 \le m < k +1$, we have
\begin{equation}
\label{eq:IH}
\sqrt{n}(\hat{\rho}_{m+1}- \rho_{m+1}) \xrightarrow{D} \N_{|\rho_{m+1}|}\Big(0, A_{m+1}(\nu)\Big).
\end{equation}
Note that for any $\alpha,\beta \in [p]$, the recursive formula for the partial correlations
\begin{equation}
\label{eq:PcorrRecursionAlphaBeta}
\rho_{\alpha\beta\cdot S_m} = \frac{\rho_{\alpha\beta\cdot S_{m+1}}-\rho_{\alpha m\cdot S_{m+1}}\rho_{\beta m\cdot S_{m+1}}}{\sqrt{1-\rho_{\alpha m\cdot S_{m+1}}^2}\sqrt{1-\rho_{\beta m\cdot S_{m+1}}^2}}\\
\end{equation}

implies that the vector $\rho_m$ can be written as a function of $\rho_{m+1}$.
Let $f_m:\R^{|\rho_{m+1}|} \rightarrow \R^{|\rho_m|}$ be this function, then we have

$$f_m(\rho_{m+1}) =
f_m\bigg(\begin{pmatrix}
\rho_{1,2 \cdot  S_{m+1}}\\
\rho_{1,3\cdot   S_{m+1}}\\
\rho_{2,3\cdot   S_{m+1}}\\
\rho_{1,4\cdot   S_{m+1}}\\
\rho_{2,4\cdot   S_{m+1}}\\
\dotsc\\ 
\rho_{1, m\cdot   S_{m+1}}\\
\rho_{2, m\cdot   S_{m+1}}\\
\dotsc\\
\rho_{{m-1},m\cdot S_{m+1}}\\
\end{pmatrix}\bigg)=
\begin{pmatrix}
\rho_{1,2 \cdot  S_{m}}\\
\rho_{1,3\cdot   S_{m}}\\
\rho_{2,3\cdot   S_{m}}\\
\rho_{1,4\cdot   S_{m}}\\
\rho_{2,4\cdot   S_{m}}\\
\dotsc\\ 
\rho_{1,(m-1)\cdot   S_{m}}\\
\rho_{2, (m-1)\cdot   S_{m}}\\
\dotsc\\
\rho_{{m-2},(m-1)\cdot S_{m}}\\
\end{pmatrix}
=\rho_m.$$
Note that this implies $f_m(\hat{\rho}_{m+1}) = \hat{\rho}_m$ since our procedure uses this recursive formula to estimate the partial correlations. 
Applying the delta method to $\eqref{eq:IH}$ with the function $f_m$ gives
\begin{equation*}
\sqrt{n}(\hat{\rho}_m - \rho_m) \xrightarrow{D} \mathcal{N}_{|\rho_{m}|}\Big(0, A_{m}(\nu, \rho_{m+1})\Big),
\end{equation*}

where $A_m(\nu, \rho_{m+1}) := \nabla f_m(\rho_{m+1})A_{m+1}(\nu)\nabla f_m(\rho_{m+1})^T$.
The matrix $\nabla f_m(\rho_{m+1})$ can be computed to be the following matrix:

$$D := \nabla f_m({\rho}_{m+1})
=
\begin{pmatrix}
a_{12}  &     0         &    0         &  \cdot &  0                            & b_{12,1} & b_{12,2}  & 0        & 0        & \cdot  & 0      \\
0             &  a_{13} &    0         &  \cdot &  0                            & b_{13,1} &   0       & b_{13,3} & 0        & \cdot  & 0      \\
0             &  0            & a_{23} &  \cdot &  0                            &  0       &  b_{23,2} & 0        & b_{23,3} & \cdot  & 0      \\
\cdot        &  \cdot       & \cdot       & \cdot  &  0                            &  \cdot  &   \cdot  &  \cdot  & \cdot   & \cdot  & \cdot \\
\cdot        &  \cdot       & \cdot       & \cdot  &  \cdot                            &  \cdot  &   \cdot  &  \cdot  & \cdot   & \cdot  & \cdot \\
0             &  0            & 0            & \cdot  &  a_{\frac{(m-1)(m)}{2}} &  \cdot  &   \cdot  &  \cdot  & \cdot   & \cdot  & \cdot \\
\end{pmatrix}
$$

where 

$$a_{xy} = \frac{1}{\sqrt{1-\rho_{x,m\cdot S_{m+1}}^2}\sqrt{1- \rho_{y,m\cdot S_{m+1}}^2} },$$

$$b_{xy,x} =\frac{ \rho_{x,y\cdot S_{m+1}}\rho_{x,m\cdot S_{m+1}} - \rho_{y,m\cdot S_{m+1}} }{\sqrt{(1-\rho_{x,m\cdot S_{m+1}}^2)^3}\sqrt{1-\rho_{y,m\cdot S_{m+1}}^2}},$$
and
$$b_{xy,y} =\frac{ \rho_{x,y\cdot S_{m+1}}\rho_{y,m\cdot S_{m+1}} - \rho_{x,m\cdot S_{m+1}} }{\sqrt{1-\rho_{x,m\cdot S_{m+1}}^2}\sqrt{(1-\rho_{y,m\cdot S_{m+1}}^2})^3}.$$

To simplify indexing, we define the index function
$$I(x,y) = x + \frac{(y-2)(y-1)}{2}.$$
Then, the element $a_{x,y}$ will be on the $I(x,y)^{th}$ row and column of the Jacobian $D$. We can now compute the elements of the matrix $A_m$ in terms of the elements of $A_{m+1}$. Namely, defining $d := \frac{(m-2)(m-1)}{2}$ and using the notation $M_{[x,y]}$ to denote the entry in the $x^{th}$ row and $y^{th}$ column of $M$, we can compute the element in the $I(x,y)^{th}$ row and $I(z,w)^{th}$ column of $A_m$ to be
\begin{align}
    \label{eq:recursion}
    A&_{m[I(x,y),I(z,w)]} = \sum_{p=1}^{I(1,m)}\sum_{q=1}^{I(1,m)} D_{m[I(x,y),p]} A_{m+1[p,q]} D^T_{m[q,I(z,w)]}\nonumber\\
    &= \sum_{p=1}^{I(1,m)}D_{m[I(x,y),p]}\sum_{q=1}^{I(1,m)} A_{m+1[p,q]} D_{m[I(z,w),q]}\nonumber\\
    &= \sum_{p=1}^{I(1,m)}D_{m[I(x,y),p]}\bigg{(}a_{z,w}A_{m+1[p,I(z,w)]} + b_{zw,z}A_{m+1[p,d + z]} + b_{zw,w}A_{m+1[p, d +  w]}\bigg{)}\nonumber\\
    &= a_{x,y}\big{(}a_{z,w}A_{m+1[I(x,y),I(z,w)]} + b_{zw,z}A_{m+1[I(x,y),d + z]} + b_{zw,w}A_{m+1[I(x,y),d+w]}\big{)}\nonumber\\
    &  + b_{xy,x}\big{(}a_{z,w}A_{m+1[d+x,I(z,w)]} + b_{zw,z}A_{m+1[d+x,d+z]} + b_{zw,w}A_{m+1[d+x,d+w]}\big{)}\nonumber\\
    &+ b_{xy,y}\big{(}a_{d+z,d+w}A_{m+1[d+y,I(z,w)]} + b_{zw,z}A_{m+1[d+y,d+z]} + b_{zw,w}A_{m+1[d+y,d+w]}\big{)}.\nonumber\\
\end{align} 

Note that equation \eqref{eq:PcorrRecursionAlphaBeta} shows that $\rho_{m+1}$ is a continuously differentiable function of $\rho_{m+2}$ since it is a composition of continuously differentiable functions. Hence, $D$ is continuous in $\rho_{m+2}$, which is continuous in $\rho$, which can be seen by arguing recursively.
Recall that we have for all $i,j\in [p]$,
\begin{equation}
\rho_{ij} = \frac{g_{ij}(\EX[\eta_{ij}])-g_i(\EX[\eta_i])g_j(\EX[\eta_j])}{\sqrt{g_{ii}(\EX[\eta_{ii}])-g_i(\EX[\eta_i])^2}\sqrt{g_{jj}(\EX[\eta_{jj}])-g_j(\EX[\eta_j])^2}}.
\end{equation}
Hence, $\rho$ is continuous in  $\eta$ and therefore in $\nu$. Finally, this implies that $D = \nabla f_m(\rho_{m+1})$ is continuous in $\nu$, and hence, so is the matrix product $A_m= DA_{m+1}(\nu)D^T$, where $A_{m+1}(\nu)$ is continuous in $\nu$ by the inductive hypothesis.
Therefore, we can reparameterize $A_m(\nu, \rho_{m+1}) = \tilde{A}_m(\nu)$, and the inductive step follows for $m \in \{3,\cdots,k\}.$
Specifically, for $m=3$, we have the desired statement:
$$\sqrt{n}(\hat{\rho}_{3} - \rho_{3}) \xrightarrow{D} \mathcal{N}_1\Big(0, \tilde{A}_3(\nu)\Big)$$
for $\tilde{A}_3(\nu)$ continuously differentiable in $\nu$. Relabeling back to $i,j$ and $K$, and defining $\tau_{ij\cdot K}(\nu) := \tilde{A}_3(\nu)$, we have
$$\sqrt{n}(\hat{\rho}_{ij\cdot K} - \rho_{ij\cdot K}) \xrightarrow{D} \N_1\Big(0, \tau_{ij\cdot K}(\nu)\Big).$$
\end{proof}

Note that it was not necessary to find the form of the elements of $A_m$ explicitly to argue that it was continuous. However, the proof of this lemma gives us a recursive formula~\eqref{eq:recursion} to compute the elements of $A_3$. Furthermore, this recursive formula is independent of the choice of noise functions $F$ and the associated functions $g$. Hence, this recursion can be used for all noise models, as long as the base case is derived for that noise model, i.e., as long as the elements of the matrix $A(\nu)$ in Lemma~\ref{lemma:CorrDelta} can be found.

\begin{proof}[Proof of Theorem \ref{thrm:AsymptoticPartialCorrelations}]
Follows directly from combining Lemma~\ref{lemma:CorrDelta} and Lemma~\ref{lemma:CorrToPcorrDelta}.
\end{proof}

\section{Proof of Corollary~\ref{cor:transform}}
\label{appendix:PoorfOfCorollaryTransform}

\begin{proof}
By Theorem~\ref{thrm:AsymptoticPartialCorrelations} and the hypothesis of the corollary,
$$\sqrt{n}(\hat{\rho}_{ij\cdot K} - \rho_{ij\cdot K}) \xrightarrow{D}\mathcal{N}_1\big(0, \tilde\tau_{ij\cdot K}(\rho_{ij\cdot K})\big).$$
An application of the delta method with
$$z_{ij\cdot K}(\rho) = \int\frac{1}{\sqrt{\tilde\tau_{ij\cdot K}(\rho)}}d\rho + C$$
gives
$$\sqrt{n}(z_{ij\cdot K}(\hat\rho_{ij\cdot K}) - z_{ij\cdot K}(\rho_{ij\cdot K}))\xrightarrow{D}\N_1\Big(0, \big(z'_{ij\cdot K}(\rho)\big)^2\tilde\tau_{ij\cdot K}(\rho_{ij\cdot K})\Big) = \N_1(0, 1).$$
Note that the condition imposed on $C$, $z_{ij\cdot K}(0) = 0$, by the corollary is not required to prove the result, but is only needed for  Theorem~\ref{thrm:AlgorithmConsistency}.
\end{proof}

\section{Proof of Corollary~\ref{cor:variance}}
\label{appendix:ProofOfCorollaryVariance}

\begin{proof}

By the law of large numbers, $\hat\nu\xrightarrow{a.s.}\nu$ as $n\rightarrow\infty$. Therefore $\tau_{ij\cdot K}( \hat\nu)\xrightarrow{a.s.}\tau_{ij\cdot K}(\nu)$ since $\tau_{ij\cdot K}$ is continuous in $\nu$ by Theorem~\ref{thrm:AsymptoticPartialCorrelations}. Combining this with the convergence result of Theorem~\ref{thrm:AsymptoticPartialCorrelations}  gives 
$$\sqrt{n}(\hat{\rho}_{i,j\cdot K} - \rho_{i,j\cdot K}) \xrightarrow{D}\N_1\big(0,\tau_{ij\cdot K}(\hat{\nu})\big)$$
and hence
$$\sqrt{n}\Big(\frac{\hat{\rho}_{i,j\cdot K}}{\sqrt{\tau_{ij\cdot K}(\hat\nu)}} - \frac{\rho_{i,j\cdot K}}{\sqrt{\tau_{ij\cdot K}(\hat\nu)}}\Big) \xrightarrow{D}\N_1\big(0,1\big).$$
If we define
$${\zeta}_{ij\cdot K}(y, \hat\nu) := \frac{y}{\sqrt{\tau_{ij\cdot K}(\hat\nu)}},$$
we obtain
$$\sqrt{n}\Big({\zeta}_{ij\cdot K}(\hat\rho_{ij\cdot K}, \hat\nu) - \zeta_{ij\cdot K}(\rho_{ij\cdot K}, \hat\nu)\Big) \xrightarrow{D}\N_1(0, 1).$$
\end{proof}
\section{Proof of Theorem~\ref{thrm:AlgorithmConsistency}}
\label{appendix:ProofOfAlgorithmConsistency}

We rely on the consistency of the causal discovery algorithm that our procedure uses such as PC~\cite{spirtes2000causation} or GSP~\cite{solus2017consistency} in the oracle setting, i.e., when the conditional independence statements of the underlying graph are known. 
Hence, to prove consistency of our procedure, it is sufficient to show that the conditional independence statements that our procedure estimates from the observed data converges to the true set of conditional independence statements under the faithfulness assumption in~\ref{asmp:a1}.

First, recall that our procedure estimates the conditional independence statements implied by $\Prob$ through declaring $X_i\indep X_j | X_K$ if and only if
\begin{equation}
\label{eq:CITestRewrite}
|T(\hat\rho_{ij\cdot K})| \le \Phi^{-1}(1 - \frac{\alpha}{2}),
\end{equation}
where $T$ could be one of two statistics: 
\begin{enumerate}
\item[(i)] $T$ is chosen as in Corollary~\ref{cor:transform} to be
\begin{equation}
\label{eq:TransformRewrite}
T(\hat\rho_{ij\cdot K}) = \sqrt{n}\,z_{ij\cdot K}(\hat\rho_{ij\cdot K}) := \sqrt{n}\bigg(\int \frac{1}{\tau_{ij\cdot K}(\hat\rho_{ij\cdot K})}d\hat\rho_{ij\cdot K} + C\bigg)
\end{equation}
with $C$ chosen such that $z_{ij\cdot K}(0) = 0$
if the conditions of Corollary~\ref{cor:transform} are satisfied,
\item[(ii)] or $T$ is chosen as in Corollary~\ref{cor:variance} to be
\begin{equation}
\label{eq:VarianceRewrite}
T(\hat\rho_{ij\cdot K}) = \sqrt{n}\,\zeta_{ij\cdot K}(\hat\rho_{ij\cdot K},\hat\nu) := \sqrt{n}\,\frac{\hat\rho_{ij\cdot K}}{\tau_{ij\cdot K}(\hat\rho_{ij\cdot K}, \hat\nu)}.
\end{equation}
\end{enumerate}

The first step in proving the theorem is the following lemma.

\begin{lemma}
\label{eq:CorrectCI}
There exists a sequence of
As $n\rightarrow\infty$, the conditional independence statements that our procedure estimates from the observations of $X$ converges to the conditional independence statements implied by $\Prob$.
\end{lemma}

\begin{proof}
Take any arbitrary $i,j\in[p]$ and $K\subseteq[p]\setminus\{i,j\}.$
First, note that in both settings of $T(\hat\rho_{ij\cdot K})$ in~\eqref{eq:TransformRewrite} and~\eqref{eq:VarianceRewrite}, $T(\hat\rho_{ij\cdot K})$ is monotonic and continuous in $\hat\rho_{ij\cdot K}.$ In the first setting it is the anti-derivative of a strictly positive function of $\hat\rho_{ij\cdot K}$ and in the second, it is linear in $\hat\rho_{ij\cdot K}$ with positive slope. Monotonicity and the definitions of $z_{ij\cdot K}$ and $\zeta_{ij\cdot K}$ imply that for $n \neq 0$, $T(\rho_{ij\cdot K}) = 0$ if and only if $\rho_{ij\cdot K} = 0$.
Continuity and Lemma~\ref{lemma:moments} imply that  
$$T(\hat\rho_{ij\cdot K})\xrightarrow{a.s.} T(\rho_{ij\cdot K})\quad\textrm{as}\quad n\rightarrow\infty.$$

Let $H_\alpha$ be the event that $X_i \nindep X_j | X_K$ was declared by the test in~\eqref{eq:CITestRewrite}. Let $H$ be the event that $X_i \nindep X_j | X_K$ according to the measure $\Prob$. Let $H'$ be the event that $X_i \indep X_j | X_K$ according to $\Prob$. We analyze the limits of the probability of declaring a CI statement correctly, $\Prob(H_{\alpha} | H)$, and the limits of declaring a CI statement incorrectly, $\Prob(H_{\alpha} | H').$
First, for all $\alpha \in (0,1]$, 
\begin{align}
\Prob(H_\alpha | H) &= \Prob(|T(\hat\rho_{ij\cdot K})| > \Phi^{-1}(1-\frac{\alpha}{2}) \,|\, \rho_{ij\cdot K} \neq 0)\nonumber\\
&\rightarrow \Prob(|T(\rho_{ij\cdot K})| > \Phi^{-1}(1-\frac{\alpha}{2}) \,|\, \rho_{ij\cdot K} \neq 0)\nonumber\\
&\rightarrow 1 \quad\textrm{as}\quad n\rightarrow\infty\label{eq:GoesToOne}
\end{align}
where to obtain~\eqref{eq:GoesToOne}, we used that $T(\rho_{ij\cdot K})\neq 0$ since $\rho_{ij\cdot K}\neq 0$. Hence, $|T(\rho_{ij\cdot K})| = |\sqrt{n}\cdot c|\rightarrow\infty$ for $c\neq 0$ as $n\rightarrow\infty$.
Moreover,
\begin{align}
\Prob(H_\alpha& | H') \\
&= \Prob\Big(|T(\hat\rho_{ij\cdot K})| > \Phi^{-1}(1-\frac{\alpha}{2}) \,\Big|\, \rho_{ij\cdot K} = 0\Big) \nonumber\\
&=  \Prob\bigg(
T(\hat\rho_{ij\cdot K}) > \Phi^{-1}(1-\frac{\alpha}{2})\bigg|T(\rho_{ij\cdot K})=0\bigg)
+
\Prob\bigg(T(\hat\rho_{ij\cdot K}) < \Phi^{-1}(\frac{\alpha}{2})  \bigg|T(\rho_{ij\cdot K}) = 0
\bigg)\nonumber\\
&\rightarrow \alpha \quad \textrm{as}\quad n\rightarrow\infty\label{eq:GoesToAlpha}
\end{align}
where~\eqref{eq:GoesToAlpha} follows from Corollaries~\ref{cor:transform} and~\ref{cor:variance} that assert the asymptotic normality of $T$ in both settings.
Hence, for any $\epsilon > 0$, we can set $\alpha_\epsilon = \epsilon/2$ and we will obtain $\Prob(H_{\alpha_\epsilon}|H') \rightarrow \alpha_\epsilon < \epsilon$ as $n\rightarrow\infty$. Therefore both errors in estimating the CI statements implied by $\Prob$ vanish asymptotically, implying that the set of CI statements obtained from observations $X$ converges to those implied by $\Prob$.
\end{proof}

\begin{proof}[Proof of Theorem 2]
Under faithfulness, the CI statements implied by $\Prob$ are those implied by $\G$. Hence, by Lemma~\ref{eq:CorrectCI}, the set of CI statements obtained from $X$ as $n\rightarrow\infty$ converge to those implied by $\G$. Therefore, if the causal discovery algorithm used in step 6 of Algorithm~\ref{alg:general} is consistent in the oracle setting, then Algorithm~\ref{alg:general} is consistent.
\end{proof}

\section{Derivation of the transforms for the dropout model}

\label{appendix:DerivationOfDropout}

In this section, we derive the transforms for the dropout model. The detailed calculations of this derivation are carried out in a supplementary Mathematica notebook, which can be found at \texttt{https://github.com/basilnsaeed/anchored\_causal\_inference}.

Recall, in the dropout model introduced in Section~\ref{section:model}, we consider an anchored causal model where $Z\sim\mathcal{N}(\mu, \Sigma)$ satisfies~\ref{asmp:a1}. In Example~\ref{ex:rna_seq}, the corrupted observation vector $X$ is modeled as
\begin{equation}
X_i = F_i(Z_i) = \begin{cases} 
Z_i & w.p \quad q_i\\
0   & w.p \quad 1-q_i
\end{cases}
\quad\textrm{for all}\quad i\in[p],
\end{equation}
with $q_i\in(0,1].$
Note that Assumption~\ref{asmp:a2} is satisfied since each $X_i$ is independent of all other variables given its parent $Z_i$. We can find the moments of $Z$ in terms of the moments of $X$:
\begin{equation}
\EX[X_i] = q_i \mu_i, \quad \EX[X_{i}^2] = q_i \mu_{ii},\quad \EX[X_iX_j] = q_{i}q_j\mu_{ij}\\
\end{equation}
for all $i,j\in[p]$ with $i\neq j$, where we defined $\mu_{ij} := \EX[Z_iZ_j].$
From this, we can see that Assumption~\ref{asmp:a2} is satisfied with
\begin{equation}
\label{eq:DropoutEta}
\eta_i := X_i,\quad \eta_{ii} := X_i^2, \quad \eta_{ij} := X_iX_j,
\end{equation}
and
\begin{equation}
\label{eq:DropoutG}
g_{i}(y) := \frac{y}{q_i},\quad g_{ii}(y):= \frac{y}{q_{i}}, \quad g_{ij}(y) := \frac{y}{q_iq_j}.
\end{equation}

\subsection{Derivation of the Dropout Stabilizing Transform}
\label{appendix:DerivationOfDropoutStabilizing}
In this section, we derive the \textit{dropout stabilizing transform} under the assumption that $\mu_i = 0$ for all $i\in[p]$ and $K = \emptyset$, i.e., we find a variance stabilizing transformation $z_{ij} = z_{ij\cdot \emptyset}$ for the correlations $\rho_{ij} = \rho_{ij\cdot \emptyset}$. 
We first show that $\tau_{ij}(\nu)$ can be reparameterized as $\tilde\tau_{ij}(\rho_{ij})$ and then solve for the dropout stabilizing transform $z_{ij}(\rho)$. We follow the proof of Lemma~\ref{lemma:CorrDelta} and later impose the $\mu = 0$ assumption.

We take any arbitrary, but fixed distinct nodes $i,j\in[p]$ and define
$$\eta := \big{(}\eta_i\; \eta_j\; \eta_{ii} \; \eta_{jj} \; \eta_{ij}  \big)^T $$
as the vector of monomials in $X_i$ and $X_j$ from~\eqref{eq:DropoutEta}.
Similarly, we define
$$\hat\eta = \big{(}\hat\eta_i\; \hat\eta_j\; \hat\eta_{ii} \; \hat\eta_{jj} \; \hat\eta_{ij}  \big)^T$$
as the analogous vector of monomials in $\hat{X_i}$ and $\hat{X_j}$
estimated from the observed data.

Then, applying the Central Limit Theorem gives
$$\sqrt{n}(\EX[\hat\eta] - \EX[\eta]) \xrightarrow{D}\N_{5}(0, A_5(\nu)),$$
where $A_5(\nu)$ is the matrix
\begin{align*}
\begin{pmatrix}
\textrm{Cov}(X_i,X_i)   & \textrm{Cov}(X_j,X_i)   & \textrm{Cov}(X_i^2,X_i)   & \textrm{Cov}(X_j^2,X_i)  & \textrm{Cov}(X_iX_j,X_i)  \\
\textrm{Cov}(X_i,X_j)   & \textrm{Cov}(X_j,X_j)   & \textrm{Cov}(X_i^2,X_j)   & \textrm{Cov}(X_j^2,X_j)  & \textrm{Cov}(X_iX_j,X_j)  \\
\textrm{Cov}(X_i,X_i^2) & \textrm{Cov}(X_j,X_i^2) & \textrm{Cov}(X_i^2,X_i^2) & \textrm{Cov}(X_j^2,X_i^2)& \textrm{Cov}(X_iX_j,X_i^2)\\
\textrm{Cov}(X_i,X_j^2) & \textrm{Cov}(X_j,X_j^2) & \textrm{Cov}(X_i^2,X_j^2) & \textrm{Cov}(X_j^2,X_j^2)& \textrm{Cov}(X_iX_j,X_j^2)\\
\textrm{Cov}(X_i,X_iX_j)  & \textrm{Cov}(X_j,X_iX_j)  & \textrm{Cov}(X_i^2,X_iX_j)  & \textrm{Cov}(X_j^2,X_iX_j) & \textrm{Cov}(X_iX_j,X_iX_j) \\
\end{pmatrix},\\
\end{align*}
and $\nu$ is the vector of all first and second order moments of $\eta$. 
Now, define $w:\R^5\rightarrow\R^1$ as 
$$w\begin{pmatrix}a_1 \\ a_2\\ a_3\\ a_4\\ a_5 \end{pmatrix} = \frac{a_5 - a_1a_2}{\sqrt{a_3 - a_1^2}\sqrt{a_4 - a_2^2}}.$$
Note that we have $w(\EX[\eta]) = \rho_{ij}$ and $w(\EX[\hat\eta]) = \hat\rho_{ij}$. 
Applying the delta method with $w$ gives
\begin{equation}
    \label{eq:DropoutCorrelationConvergence}
    \sqrt{n}(\hat\rho_{ij} - \rho_{ij}) \xrightarrow{D}\N_1(0, \tau_{ij}(\nu)),
\end{equation}
where
\begin{equation}
\label{eq:AsymptoticVarNu}
\tau_{ij}(\nu) = \nabla w(\EX[\eta])^T A_5(\nu)\nabla w(\EX[\eta]).
\end{equation}
Carrying out the multiplication gives the asymptotic variance $\tau_{ij}(\nu)$ parameterized by elements of $
\nu$. In the case of the dropout model, any moments of $X$ are linear in moments of $Z$, for example,
$$\EX[X_iX_k^2X_j] = q_iq_kq_j \EX[Z_iZ_kZ_j].$$ Furthermore, any moments of $Z$, which is a Gaussian random variable, can be written as polynomials in the first and second order moments of $Z$, i.e., the elements of $\mu$ and $\Sigma$. Hence, after imposing the constraint that $\mu = 0$, we can reparameterize $\tau_{ij}(\nu)$ in terms of $\Sigma$ as

$$\bar\tau_{ij}(\Sigma) =  \frac{1}{q_iq_j} + \frac{2(\frac{\sigma_{ij}}{\sqrt{\sigma_{ii}}\sqrt{\sigma_{jj}}})^2}{q_iq_j} - \frac{9(\frac{\sigma_{ij}}{\sqrt{\sigma_{ii}}\sqrt{\sigma_{jj}}})^2}{4q_j} - \frac{9(\frac{\sigma_{ij}}{\sqrt{\sigma_{ii}}\sqrt{\sigma_{jj}}})^2}{4q_i} + \frac{(\frac{\sigma_{ij}}{\sqrt{\sigma_{ii}}\sqrt{\sigma_{jj}}})^2}{2} + (\frac{\sigma_{ij}}{\sqrt{\sigma_{ii}}\sqrt{\sigma_{jj}}})^4,$$
where $\sigma_{kl} = (\Sigma)_{kl}$. The details of the computation are included in the supplementary Mathematica notebook.
Now, using $\rho_{ij} = \frac{\sigma_{ij}}{\sqrt{\sigma_{ii}}\sqrt{\sigma_{jj}}}$ we can reparameterize $\bar\tau_{ij}(\Sigma)$ once more to obtain
$$\tilde\tau_{ij}(\rho_{ij}) = \frac{1}{q_iq_j} + \frac{2\rho_{ij}^2}{q_iq_j} - \frac{9\rho_{ij}^2}{4q_j} - \frac{9\rho_{ij}^2}{4q_i} + \frac{\rho_{ij}^2}{2} + \rho_{ij}^4.$$
Hence, in the $\mu=0$ case, we can rewrite~\eqref{eq:DropoutCorrelationConvergence} as
$$\sqrt{n}(\hat\rho_{ij} - \rho_{ij}) \xrightarrow{D}\N_1(0, \tilde\tau_{ij}(\rho_{ij})).$$

In order to find a variance stabilizing transform for $\rho_{ij}$, we can now solve 
$$z_{ij}(\rho) = \int \frac{1}{\sqrt{\tilde\tau_{ij}(\rho)}} d\rho + C$$
with $C$ chosen such that $z_{ij}(0) = 0$.
Then, by Corollary~\ref{cor:transform}, we will have
$$\sqrt{n}\Big(z_{ij}(\hat\rho_{ij}) - z_{ij}(\rho_{ij})\Big) \xrightarrow{D}\N_1(0, 1).$$

There is no closed form for $z_{ij}(\rho)$ in this case. However, it can written as
\begin{equation}
\label{eq:dropoutTransform}
z_{ij}(\rho) = -\textbf{i} \frac{\sqrt{1 - \frac{8q_iq_j\rho^2}{z_+}}\sqrt{1 + \frac{8q_iq_j\rho^2}{z_-}} \mathlarger{\mathlarger{\int}}_0^{\textbf{i}\arcsin{2\rho\sqrt{\frac{2q_iq_j}{z_+}}}} {({1 + \frac{z_+}{z_-}\sin^2{\theta}})^{-\frac12}d\theta}}  
{\sqrt{(\frac{2q_iq_j}{z_+})  (\frac{4 + (8-9q_i-9q_j + 2q_iq_j)\rho^2  + 4 q_iq_j\rho^4}{q_iq_j})} },\\
\end{equation}
where $\textbf{i} = \sqrt{-1},$ and
\begin{align*}
z_+ &= +8 - 9q_i  - 9q_j  + 2q_iq_j + \sqrt{-64q_iq_j + (8 - 9q_i -9q_j + 2q_iq_j)^2},\\
z_- &= -8 + 9q_i  + 9q_j  - 2q_iq_j + \sqrt{-64q_iq_j + (8 - 9q_i -9q_j + 2q_iq_j)^2}.\\
\end{align*}
The integral that appears in the expression of $z_{ij}(\rho)$ is the elliptic integral of the first kind, and can be computed numerically.

\subsection{Conditions for the Dropout Stabilizing Transform}
\label{appendix:DropoutStabilizingTransformConditionPlots}
As mentioned in Section~\ref{section:implementation}, the dropout stabilizing transform only exists when $\mu=0$ and $K = \emptyset$.
If the derivation was done with non-zero means, it would not have been possible to reparameterize the asymptotic variance of the correlations $\tau_{ij}(\nu)$ in terms of only the correlation $\rho_{ij}$ to satisfy the conditions of Corollary~\ref{cor:transform}. In Figure~\ref{fig:TauOfMu}, we demonstrate the dependence of $\tau_{ij}(\nu)$ from equation \eqref{eq:AsymptoticVarNu} for fixed $\rho_{ij}$ on $\sigma:=\sigma_{ii} =\sigma_{jj}$, which are elements of $\nu$, when $\mu\neq 0$. This can be additionally verified through the supplementary Mathematica notebook. 
Figures~\ref{fig:TauOfMu}(a,b,d,e) show that $\tau_{ij}$ is still dependent on elements of $\nu$, even for a fixed correlation $\rho_{ij}$ for $q\neq 1$ when $\mu\neq 0$, and hence a transform of the kind in Corollary~\ref{cor:transform} does not exist for $\mu \neq 0$. For $q=1$, i.e.~no dropout, the dropout model reduces to the measurement-error-free Gaussian, and $\tau_{ij}$ no longer depends on $\mu$ and $\sigma$ for fixed $\rho_{ij}$. In this case, a transform of the kind in Corollary~\ref{cor:transform} does exist and as shown in~\cite{lehmann}, it is the Fisher z-transform.

\begin{figure}[H]
	\centering
	\subfigure[$\mu = 2$]{\includegraphics[width=.32\linewidth]{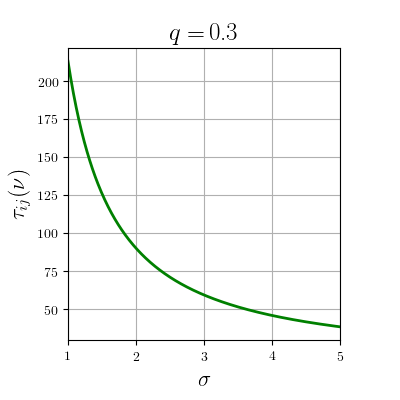}}
	\subfigure[$\mu = 2$]{\includegraphics[width=.32\linewidth]{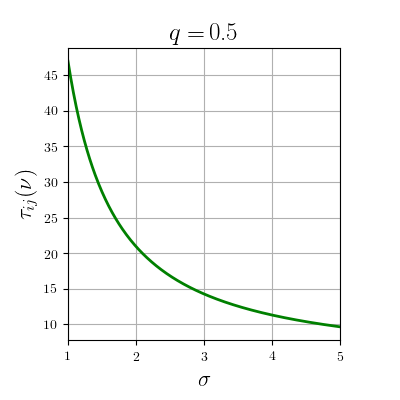}}
	\subfigure[$\mu = 2$]{\includegraphics[width=.32\linewidth]{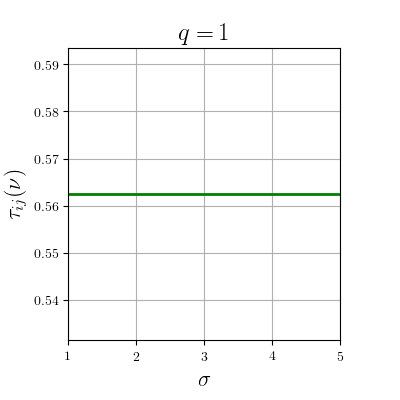}} \\

	\subfigure[$\mu = 1$]{\includegraphics[width=.32\linewidth]{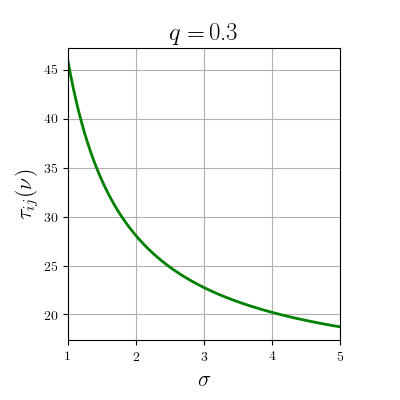}}
	\subfigure[$\mu = 1$]{\includegraphics[width=.32\linewidth]{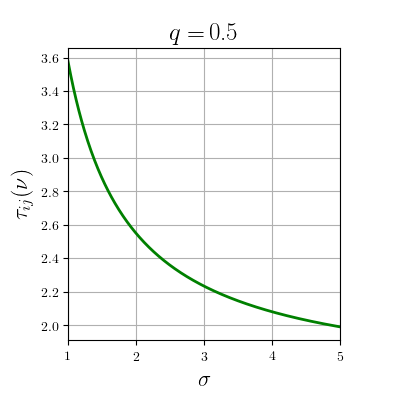}}
	\subfigure[$\mu = 1$]{\includegraphics[width=.32\linewidth]{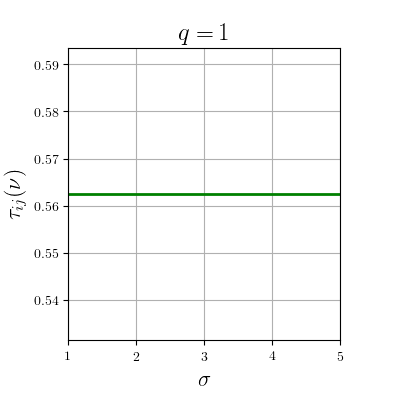}} \\
	
	\subfigure[$\mu = 0$]{\includegraphics[width=.32\linewidth]{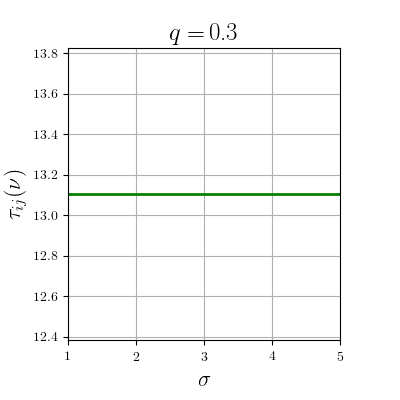}}
	\subfigure[$\mu = 0$]{\includegraphics[width=.32\linewidth]{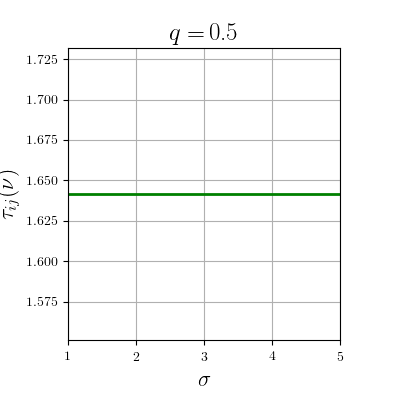}}
	\subfigure[$\mu = 0$]{\includegraphics[width=.32\linewidth]{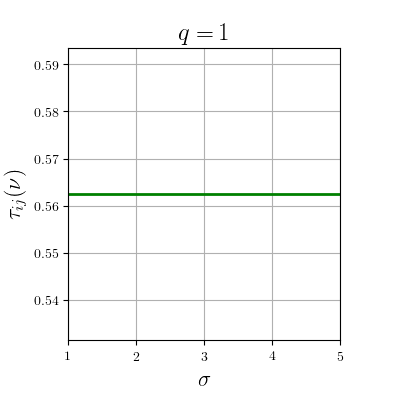}}
	\caption{Plots of $\tau_{ij\cdot\emptyset}$ when $\rho_{ij}$ is fixed to $0.5$, $\mu\in\{0,1,2\}$, with $\sigma$ allowed to vary. This shows that we cannot reparameterize $\tau_{ij}$ as a function of only $\rho_{ij}$ for non-zero mean, unless $q=1$. For $q=1$ the transform corresponding to Corollary~\ref{cor:transform} is the Fisher's z-transform.}
	\label{fig:TauOfMu}
	\vspace{-0.3cm}
\end{figure}

\subsection{Derivation of the Dropout Normalizing Transform}

\label{appendix:DerivationOfNormalizing}
In this section we give a way to compute the \textit{dropout normalizing transform} corresponding to Corollary~\ref{cor:variance} under the dropout model. We begin by showing how to compute the asymptotic variance of the partial correlations, $\tau_{ij\cdot K}(\nu)$.


In the proof of Lemma~\ref{lemma:CorrToPcorrDelta}, we showed that if we know the continuous function $A_\rho(\nu)$ such that
\begin{equation}
    \sqrt{n}(\hat\rho - \rho) \xrightarrow{D}\N_{|\rho|}\Big(0, A_\rho(\nu)\Big),
\end{equation}
then we can recursively compute the function $\tau_{ij\cdot K}(\nu)$ beginning with the matrix $A_\rho(\nu)$. Hence, to give a way to compute $\tau_{ij\cdot K}$ for the dropout model, it is sufficient to describe the elements of the  matrix $A_{\rho}(\nu)$ and thus we find a formula for each element of the $A_{\rho}(\nu)$ matrix. First, recall that the elements of $A_\rho(\nu)$ correspond to the covariances of the sample correlations of the latent variables $Z$ estimated in step 3 of Algorithm~\ref{alg:general}. That is, each element of $A_{\rho}(\nu)$ will correspond to the asymptotic covariance of $\sqrt{n}\hat\rho_{ab}$ and $\sqrt{n}\hat\rho_{cd}$ for some $a,b,c,d\in[p]$ such that $a\neq b$ and $c\neq d$.
There are three different cases for each entry in $A_{\rho}(\nu)$,
corresponding to different cases of $a,b,c,d$:
\begin{enumerate}
\item[(i)] $\{a,b\} = \{c,d\}$ are distinct, and the element is along the diagonal, corresponding to the asymptotic variance of $\sqrt{n}\hat\rho_{ab}$,
\item[(ii)] $a\not\in\{c,d\}$ and $b\in\{c,d\}$,
\item[(iii)] all of $a,b,c,d$ are distinct.
\end{enumerate}
To analyze all three cases, it is sufficient to take four arbitrary, but fixed distinct $i,j,k,l\in[p]$. We begin by noting that for the dropout model, we can write
\begin{align}
\label{eq:corrFromMoments}
\hat{\rho}_{ij}
&=\frac{\hat{\mu}_{ij} - \hat\mu_i\hat\mu_j}{\sqrt{\hat\mu_{ii}-\hat\mu_i^2 } \sqrt{\hat\mu_{jj} - \hat\mu_j^2 } }\nonumber\\
&= \frac{\frac{1}{q_iq_j}\EX[\hat\eta_{ij}]-\frac{1}{q_i}\EX[\hat\eta_i]\frac{1}{q_j}\EX[\hat\eta_j]}
{
\sqrt{\frac{1}{q_i}\EX[\hat{\eta}_{ii}]-(\frac{1}{q_i}\EX[\hat{\eta}_i])^2}
\sqrt{\frac{1}{q_j}\EX[\hat{\eta}_{jj}]-(\frac{1}{q_j}\EX[\hat{\eta}_j])^2}
}\nonumber\\
&= \frac{\EX[\hat{\eta}_{ij}]-\EX[\hat{\eta}_i]\EX[\hat{\eta}_j]}{\sqrt{q_i\EX[\hat{\eta}_{ii}]-\EX[\hat{\eta}_i]^2}\sqrt{q_j\EX[\hat{\eta}_{jj}]-\EX[\hat{\eta}_j]^2}}.\nonumber\\
\end{align}

Define 
$$\hat{\rho} = 
(\hat{\rho}_{ij}\;
\hat{\rho}_{ik}\;
\hat{\rho}_{il}\;
\hat{\rho}_{jk}\;
\hat{\rho}_{jl}\;
\hat{\rho}_{kl})^T,
$$
as the vector of estimated correlations of $X_i,X_j,X_k,X_l$ obtained from $\EX[\hat\eta]$ by~\eqref{eq:corrFromMoments}. Similarly, let
$$\rho = (\rho_{ij}\;
\rho_{ik}\;
\rho_{il}\;
\rho_{jk}\;
\rho_{jl}\;
\rho_{kl})^T$$
be the analogous vector of true correlations. In the next part of the derivation, we will apply the delta method to the vectors of moments of the monomials in $X_i$, $X_j,X_k,$ and $X_l$ of Assumption~\ref{asmp:a2}, to obtain the asymptotic distribution of the vector $\hat\rho$, as in the proof of Lemma~\ref{lemma:CorrDelta},  We begin by defining the vector of relevant monomials in $X_i, X_j, X_k,$ and $X_l$ as
$$\eta = \big{(}\eta_i\; \eta_j\; \eta_k \;  \eta_l \; \eta_{ii} \; \eta_{ij} \; \eta_{ik}\;  \eta_{il}  \; \eta_{jj}  \; \eta_{jk}  \; \eta_{jl} \; \eta_{kk} \; \eta_{kl} \; \eta_{ll}\big{)}^T$$
where the components are defined for our model in equation~\eqref{eq:DropoutEta}. Then, by the Central Limit Theorem, we have
\begin{equation}
    \sqrt{n}(\EX[\hat\eta] - \EX[\eta])\xrightarrow{D}\N_{14}\Big(0, A_{14}(\nu)\Big)
\end{equation}
where $A_{14}(\nu)$ is the covariance matrix of the vector $\eta$, and $\nu$ is the vector of all first and second order moments of $\eta$.
To obtain the convergence result stated in the Lemma~\ref{lemma:CorrDelta}, we define the function $w_\rho:\mathbb{R}^{14}\rightarrow\mathbb{R}^6$ based on \eqref{eq:corrFromMoments} such that
$w_\rho(\EX[\eta]) = \rho$ and $w_{\rho}(\EX[\hat\eta]) = \hat\rho$. 
Then 
\begin{equation}
\label{eq:DeltaWRho}
\sqrt{n}(\hat{\rho} - \rho) \xrightarrow{D}\mathcal{N}_6\Big(0,\nabla w_\rho(\EX[\eta])^{T}A_{14}(\nu)\nabla w_\rho(\EX[\eta])\Big).
\end{equation}

Since the moments in $\eta$ are included in the vector of moments $\nu$,  we can define 
\begin{equation}
\label{eq:MatrixMul}
A_6(\nu):=\nabla w_\rho(\EX[\eta])^{T}A_{14}(\nu)\nabla w_\rho(\EX[\eta]).
\end{equation}

The explicit form of $A_6(\nu)$ can be found by carrying out the matrix multiplication in \eqref{eq:MatrixMul}. Before performing the matrix multiplication, we note that for the dropout model, we can write any moment of $X$ as a linear function of a moment of $Z$, for example,
$$\EX[X_iX_k^2X_j] = q_iq_kq_j \EX[Z_iZ_kZ_j].$$ 
Furthermore, since $Z$ is a Gaussian random vector and all moments of a Gaussian random vector can be written in terms of its first and second order moments, we can parameterize the asymptotic covariance with the moments of the Gaussian $Z$ as
$$\bar{A}(\mu,\Sigma) := A_6(\eta).$$

For each entry in $A_6(\eta)$, we list the three cases mentioned previously in terms of the parameterization as $\bar{A}(\mu,\Sigma)$. The full computation is carried out in the supplementary Mathematica notebook. We use the notation $\sigma_{ij}:= (\Sigma)_{ij}$
to denote the elements of $\Sigma$. For $a,b,c,d \in \{i,j,k,l\}$, 
\begin{enumerate}
    \item[(i)] If the element corresponds to the asymptotic covariance of $\sqrt{n}\hat\rho_{ab}$ and $\sqrt{n}\hat\rho_{cd}$ with $\{a,b\} = \{c,d\}$, then it is equal to
  \begin{align*}
&\frac{1}{\sigma_{aa}\sigma_{bb}}\big{(} \frac{\sigma_{aa}\sigma_{bb}}{q_aq_b}\\
&+(-\mu_a^2\sigma_{bb} -\mu_a^2\mu_b^2 -4\mu_a\mu_b\sigma_{ab}+\frac{\mu_b^4\sigma_{ab}^2}{4\sigma_{bb}^2} + \frac{\mu_{b}^2\sigma_{ab}^2}{2\sigma_{bb}})\frac{1}{q_b}\\
&+(-\mu_b^2\sigma_{aa} -\mu_a^2\mu_b^2 -4\mu_a\mu_b\sigma_{ab}+\frac{\mu_a^4\sigma_{ab}^2}{4\sigma_{aa}^2} + \frac{\mu_{a}^2\sigma_{ab}^2}{2\sigma_{aa}})\frac{1}{q_a}\\
&+(\mu_a^2\sigma_{bb}+\mu_b^2\sigma_{aa}+\mu_a^2\mu_b^2+4\mu_a\mu_b\sigma_{ab}+2\sigma_{ab}^2)\frac{1}{q_aq_b}\\
&+(-\frac{\mu_a^4 \sigma_{ab}^2}{4\sigma_{aa}^2}-\frac{\mu_a^2\sigma_{ab}^2}{2\sigma_{aa}}-\frac{\mu_b^4 \sigma_{ab}^2}{4\sigma_{bb}^2}-\frac{\mu_b^2\sigma_{ab}^2}{2\sigma_{bb}} +\frac{\sigma_{ab}^4}{\sigma_{aa}\sigma_{bb}} + \mu_a^2\mu_b^2 + 4\mu_a\mu_b\sigma_{ab} + \frac{\sigma_{ab}^2}{2})\\
&-\frac{9\sigma_{ab}^2}{4}(\frac{1}{q_a}+\frac{1}{q_b})\big{)}\\
  \end{align*}

\item[(ii)] If the element corresponds to the asymptotic covariance of $\sqrt{n}\hat\rho_{ab}$ and $\sqrt{n}\hat\rho_{cd}$ with $a\not\in\{c,d\}$ and $b=d$, then the element is equal to
\begin{align*}
        &\bigg{(}2q_b\sigma_{bb}\sigma_{bc}(\sigma_{ab}\sigma_{ac}^2\sigma_{bb}-2\sigma_{aa}\sigma_{ac}\sigma_{bb}\sigma_{bc} + \sigma_{aa}\sigma_{ab}\sigma_{bc}^2)\\
        &- \sigma_{cc}(1 - q_b)\mu_{b}^4\sigma_{aa}\sigma_{ab}\sigma_{bc}\\
        &+2 (1 - q_b) \mu_{b}^2\sigma_{aa}\sigma_{bb}\sigma_{cc}\sigma_{ab}\sigma_{bc}\\
        &- (1 + q_b) \sigma_{bb}^2\sigma_{cc}\sigma_{aa}\sigma_{ab}\sigma_{bc}\\
        &+\sigma_{cc}\sigma_{bb}^2\sigma_{aa}(\sigma_{ab}\sigma_{bc} - 4\sigma_{ac}\sigma_{bb})\\
        & - 4 \sigma_{bb}^2\sigma_{ac}\sigma_{cc}(q_b\sigma_{ab}^2 + (q_b-1)\mu_b^2\sigma_{aa} + \sigma_{bb}\sigma_{aa})\bigg)\frac{1}{4q_b\sigma_{bb}^3\sqrt{(\sigma_{aa}\sigma_{cc})^3}}
\end{align*}

\item[(iii)] If the element corresponds to the asymptotic covariance of $\sqrt{n}\hat\rho_{ab}$ and $\sqrt{n}\hat\rho_{cd}$ with $a,b,c,d$ distinct, then the element is equal to
\begin{align*}
   &\bigg{(}\sigma_{cc}\sigma_{cd}(\sigma_{ab}\sigma_{ad}^2\sigma_{bb}- 2\sigma_{aa}\sigma_{ad}\sigma_{bb}\sigma_{bd}+\sigma_{aa}\sigma_{ab}\sigma_{bd}^2)\\
   &+ \sigma_{aa}\sigma_{bb}(\sigma_{ad}\sigma_{bc}\sigma_{cc} + \sigma_{ac}\sigma_{bd}\sigma_{cc} - \sigma_{ac}\sigma_{bc}\sigma_{cc} - \sigma_{ad}\sigma_{bc}\sigma_{cd})\sigma_{dd}\\
   &- 2\sigma_{ac}\sigma_{ad}\sigma_{bb]}\sigma_{cc}\sigma_{dd}\sigma_{ab}\\
   &+ \sigma_{ac}^2\sigma_{bb}\sigma_{cd}\sigma_{dd}\sigma_{ab}\\
   &+\sigma_{aa}\sigma_{bc}\sigma_{dd}\sigma_{ab}(\sigma_{bc}\sigma_{cd}-2\sigma_{bd}\sigma_{cc})\bigg{)}\frac{1}{2\sqrt{(\sigma_{aa}\sigma_{bb}\sigma_{cc}\sigma_{dd})^3}}
\end{align*}
\end{enumerate}

These expressions can be used as the base case of the recursion~\eqref{eq:recursion}, which allows us to indirectly find $\tau_{ij\cdot K}$ for any $i,j \in[p]$ and $K\subseteq[p]\setminus\{i,j\}$ under the dropout model. 

The dropout normalizing transform is then computed using
$$ \zeta_{ij\cdot K}(\hat{\rho}, \hat\nu) := {\hat{\rho}}/{\sqrt{\tau_{ij\cdot K}({\hat\nu})}}.$$

\subsection{Derivation of the Dropout Normalizing Transform with Shrinkage}

\label{appendix:DerivationOfNormalizingShrinkage}
In this section, we derive the dropout normalizing transform when the partial correlations are estimated from the shrinkage matrix $\hat\Lambda$ for a fixed shrinkage coefficient $\alpha$. The derivation closely follows Section~\ref{appendix:DerivationOfNormalizing}. First recall that from Section~\ref{section:implementation}
$$\hat\Lambda = (1-\alpha)\hat\Sigma  + \alpha\hat{S}$$ 
where $\hat{S} = \sum_{i=1}^nX^{(i)}X^{(i)T}$ is covariance matrix of the observations of $X$.
Then, denoting $\hat\lambda_{ij} =(\hat\Lambda)_{ij},$ we can express correlations as
\begin{align}
\label{eq:LambdaCorr}
\hat\rho_{ij} &=
\frac{\hat\lambda_{ij}}{\sqrt{\hat\lambda_{ii\phantom{j}}}\sqrt{\hat\lambda_{jj}}}\nonumber\\
\end{align}
where the elements of $\hat\Lambda$ are
\begin{align}
\label{eq:EtaToLambda}
\hat\lambda_{ij} &= (1-\alpha)\Big(\frac{1}{q_iq_j}\EX[\hat{\eta}_{ij}]-\frac{1}{q_i}\EX[\hat{\eta}_i]\frac{1}{q_j}\EX[\hat{\eta}_j]\Big)+\alpha\Big(\EX[\hat\eta_{ij}] -\EX[\hat\eta_i]\EX[\hat\eta_j]\Big)\nonumber,\\
\hat\lambda_{ii} &= (1-\alpha)\Big(\frac{1}{q_i}\EX[\hat{\eta}_{ii}]-(\frac{1}{q_i}\EX[\hat{\eta}_i])^2\Big) + \alpha\Big(\EX[\hat\eta_{ii}] - \EX[\hat\eta_i]^2\Big)\nonumber,\\
\hat\lambda_{jj} &= (1-\alpha)\Big(\frac{1}{q_j}\EX[\hat{\eta}_{jj}]-(\frac{1}{q_j}\EX[\hat{\eta}_j])^2\Big) + \alpha\Big(\EX[\hat\eta_{jj}] - \EX[\hat\eta_j]^2\Big) \nonumber.\\
\end{align}
We can define the function $w_{\rho}$ of equation~\eqref{eq:DeltaWRho} based on equations~\eqref{eq:LambdaCorr} and~\eqref{eq:EtaToLambda} and proceed as in Section~\ref{appendix:DerivationOfNormalizing} to derive the corresponding elements of the asymptotic covariance matrix. The derivation of the dropout normalizing transform with shrinkage, in addition to the result is shown in the supplementary Mathematica notebook. Note that in this case, the elements of the asymptotic covariance matrix of $\sqrt{n}\hat\rho$ will be functions of $\alpha$.


\section{Experiments}
\label{appendix:MorePlots}
We include additional simulation results for varying $p\in\{10, 30\}$, $n \in \{1000, 2000, 10000, 50000\}$ and $d \in \{2, 3, 5\}$. Specifically, we evaluate the estimated skeleton as well as the CPDAG in recapitulating the true DAG $\G$ using ROC curves and SHD. For the majority of the settings, the dropout stabilizing transform outperforms the naive Gaussian CI test applied on the corrupted data. As pointed out in Section~\ref{section:simulations}, both dropout transforms tend to outperform the Gaussian CI test when the number of samples is high. In plotting the ROC curve for the CPDAG for $p\in\{10, 30\}$, we consider an undirected edge in the CPDAG a true positive if a directed edge exists in $\G$ in either direction, and a false positive otherwise. We consider a directed edge in the CPDAG a true positive if a directed edge of the same direction exists in $\G$, and a false positive otherwise.

We also include the inferred gene regulatory network for the pancreatic type II diabetes data set, collected with inDrop single-cell RNA-seq technology. We use the dropout stabilizing transform and Algorithm~\ref{alg:general} to obtain causal relationships between latent genes. \newpage

\begin{figure}[H]
	\centering
	\vspace{-0.3cm}
	\subfigure[]{\includegraphics[width=.24\linewidth]{figures/SKELp10n1000d3.png}}
	\subfigure[]{\includegraphics[width=.24\linewidth]{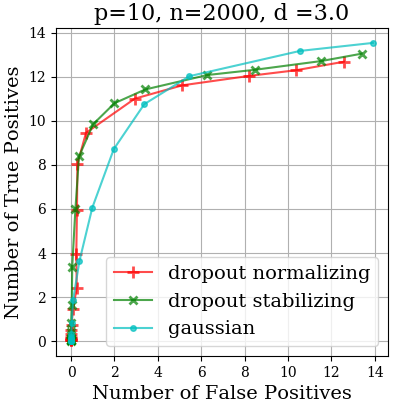}}
	\subfigure[]{\includegraphics[width=.24\linewidth]{figures/SKELp10n10000d3.png}} 
	\subfigure[]{\includegraphics[width=.24\linewidth]{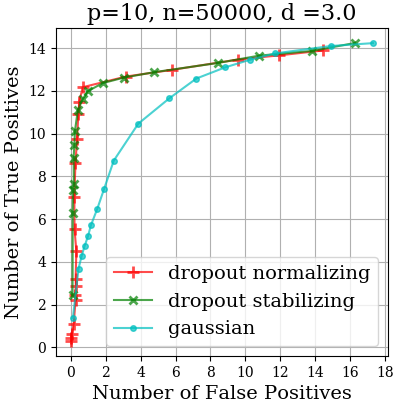}} \\
	\subfigure[]{\includegraphics[width=.24\linewidth]{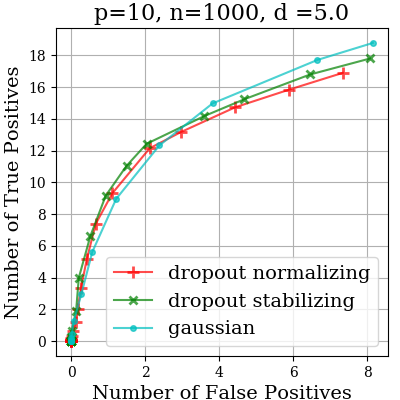}}
	\subfigure[]{\includegraphics[width=.24\linewidth]{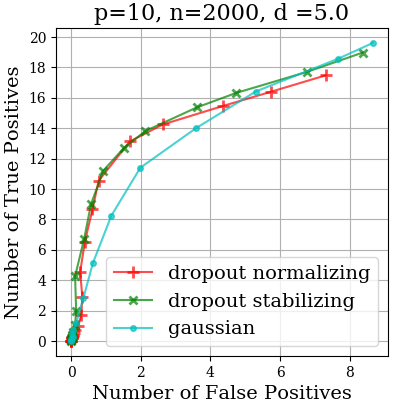}}
	\subfigure[]{\includegraphics[width=.24\linewidth]{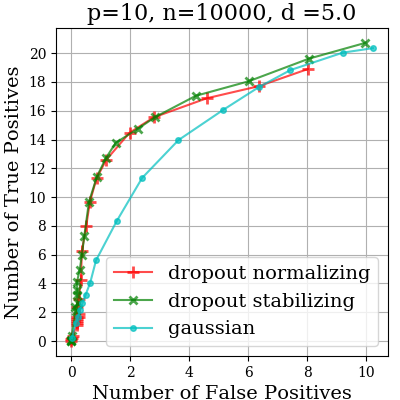}} 
	\subfigure[]{\includegraphics[width=.24\linewidth]{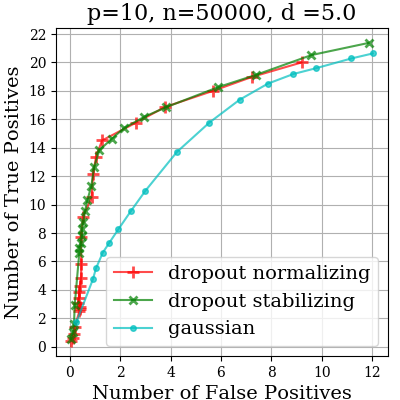}} 
	\caption{ROC curves for evaluating the estimated skeleton of the true DAG using dropout stabilizing transform, dropout normalizing transform, and Gaussian CI test in simulations with $p=10$ and $n \in \{1000, 2000, 10000, 50000\}$ and $d \in \{3, 5\}$.}
\end{figure}

\begin{figure}[H]
	\centering
	\vspace{-0.3cm}
	\subfigure[]{\includegraphics[width=.24\linewidth]{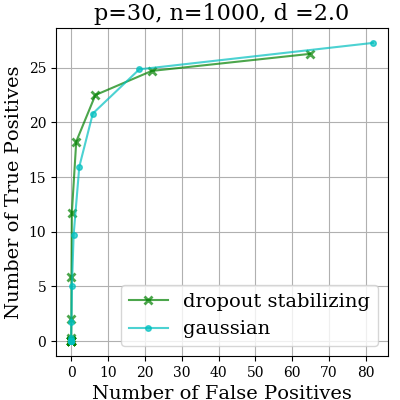}}
	\subfigure[]{\includegraphics[width=.24\linewidth]{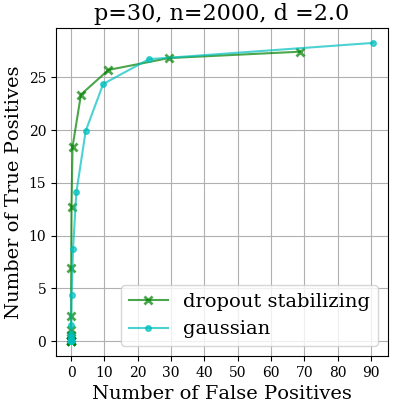}}
	\subfigure[]{\includegraphics[width=.24\linewidth]{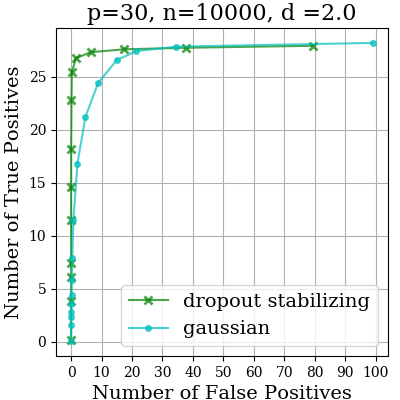}} 
	\subfigure[]{\includegraphics[width=.24\linewidth]{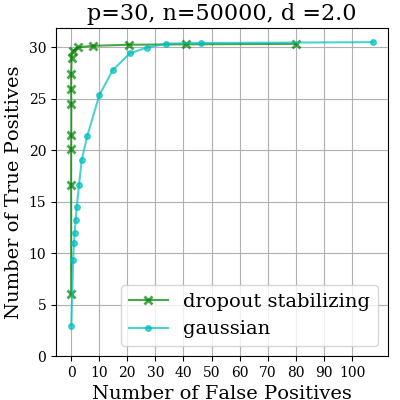}} \\
	\subfigure[]{\includegraphics[width=.24\linewidth]{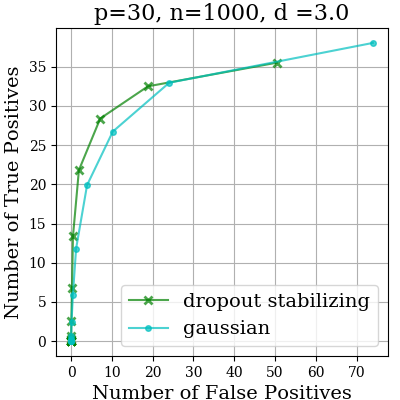}}
	\subfigure[]{\includegraphics[width=.24\linewidth]{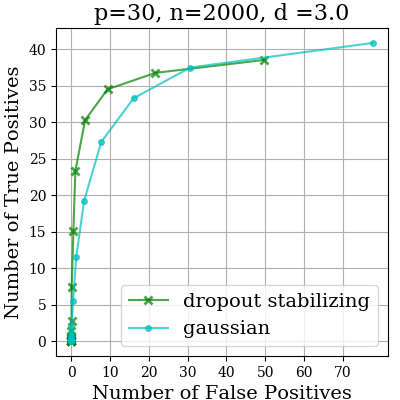}}
	\subfigure[]{\includegraphics[width=.24\linewidth]{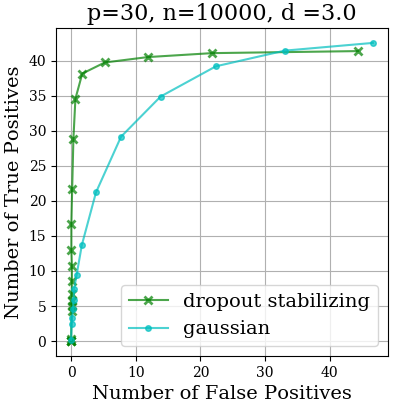}} 
	\subfigure[]{\includegraphics[width=.24\linewidth]{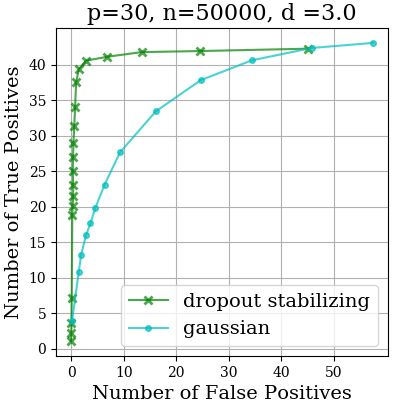}} 
	\caption{ROC curves for evaluating the estimated skeleton of the true DAG using dropout stabilizing transform and Gaussian CI test in simulations with $p=30$ and $n \in \{1000, 2000, 10000, 50000\}$ and $d \in \{2, 3\}$.}
\end{figure}

\begin{figure}[H]
	\centering
	\vspace{-0.3cm}
	\subfigure[]{\includegraphics[width=.24\linewidth]{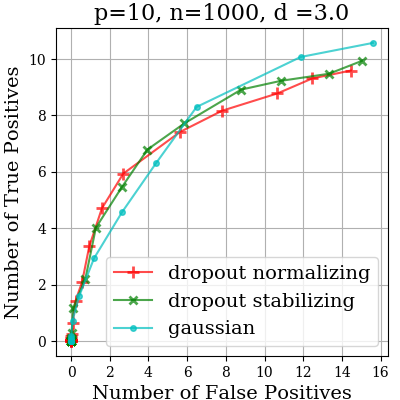}}
	\subfigure[]{\includegraphics[width=.24\linewidth]{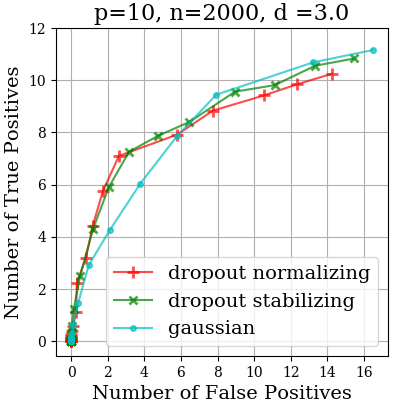}}
	\subfigure[]{\includegraphics[width=.24\linewidth]{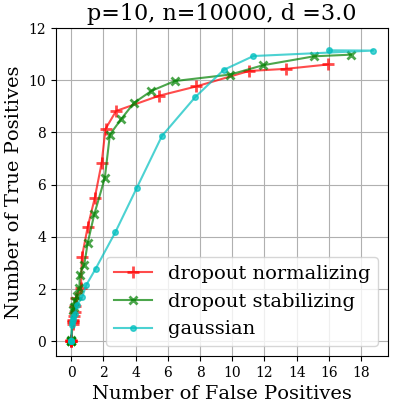}} 
	\subfigure[]{\includegraphics[width=.24\linewidth]{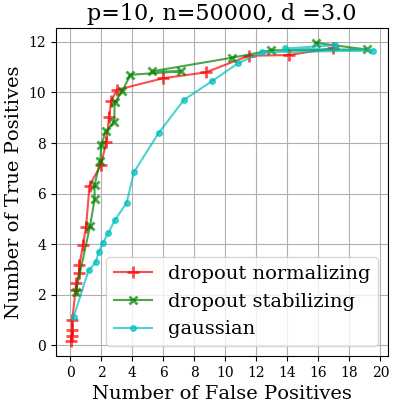}} \\
	\subfigure[]{\includegraphics[width=.24\linewidth]{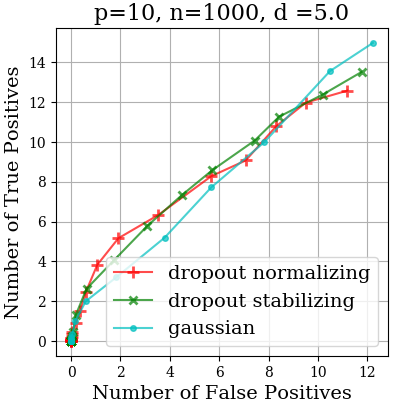}}
	\subfigure[]{\includegraphics[width=.24\linewidth]{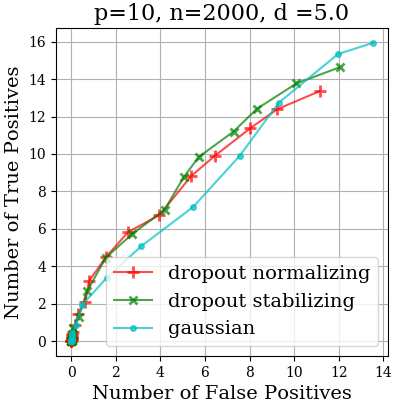}}
	\subfigure[]{\includegraphics[width=.24\linewidth]{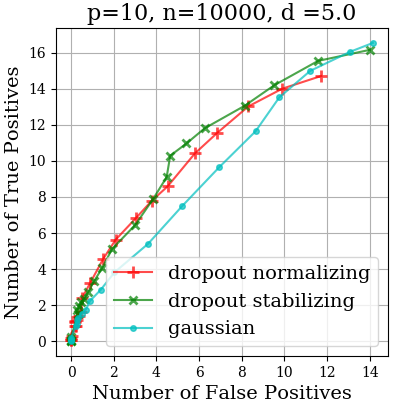}} 
	\subfigure[]{\includegraphics[width=.24\linewidth]{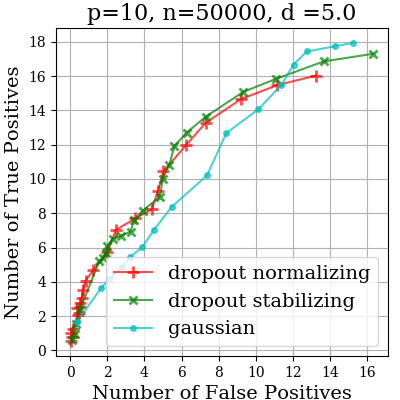}} 
	\caption{ROC curves for evaluating the estimated CPDAG of the true DAG using dropout stabilizing transform, dropout normalizing transform, and Gaussian CI test in simulations with $p=10$ and $n \in \{1000, 2000, 10000, 50000\}$ and $d \in \{3, 5\}$.}
\end{figure}

\begin{figure}[H]
	\centering
	\vspace{-0.3cm}
	\subfigure[]{\includegraphics[width=.24\linewidth]{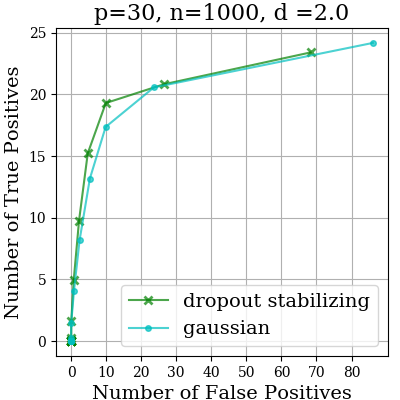}}
	\subfigure[]{\includegraphics[width=.24\linewidth]{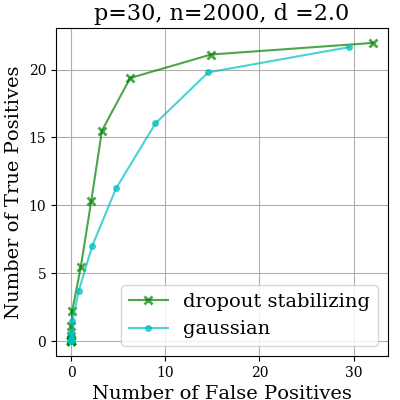}}
	\subfigure[]{\includegraphics[width=.24\linewidth]{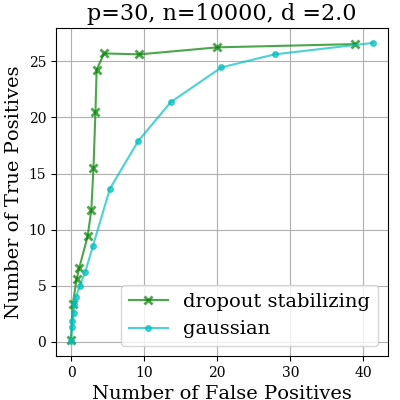}} 
	\subfigure[]{\includegraphics[width=.24\linewidth]{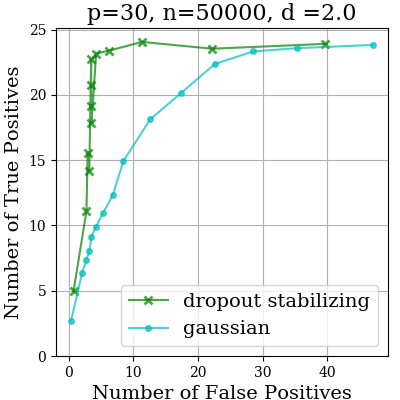}}\\
	\subfigure[]{\includegraphics[width=.24\linewidth]{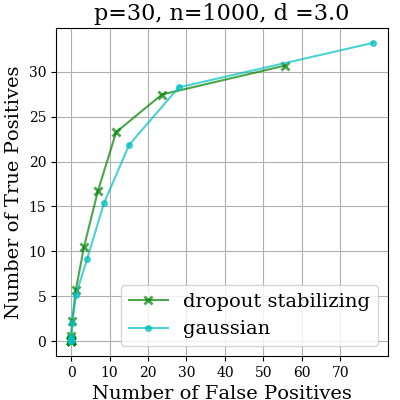}}
	\subfigure[]{\includegraphics[width=.24\linewidth]{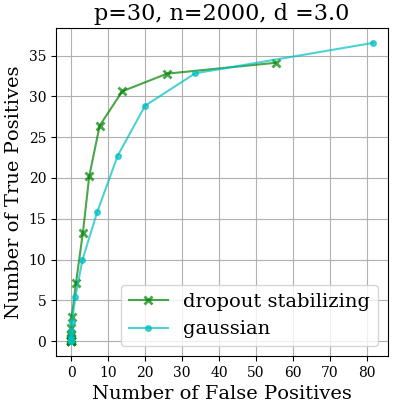}}
	\subfigure[]{\includegraphics[width=.24\linewidth]{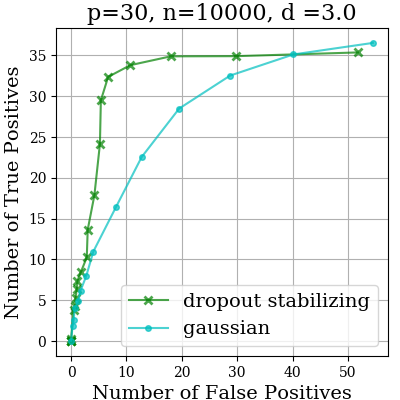}} 
	\subfigure[]{\includegraphics[width=.24\linewidth]{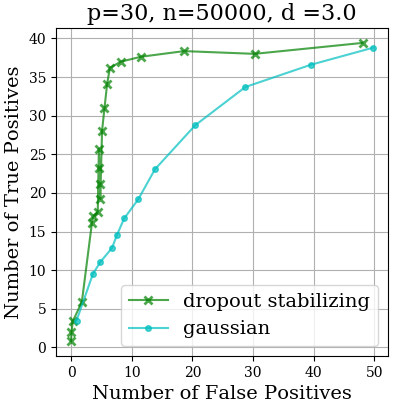}} 
	\caption{ROC curves for evaluating the estimated CPDAG of the true DAG using dropout stabilizing transform and Gaussian CI test in simulations with $p=30$ and $n \in \{1000, 2000, 10000, 50000\}$ and $d \in \{2, 3\}$.}
\end{figure}

\begin{figure}[H]
	\centering
	\vspace{-0.3cm}
	\subfigure[]{\includegraphics[width=.24\linewidth]{figures/SHDSKELp10n1000d3.png}}
	\subfigure[]{\includegraphics[width=.24\linewidth]{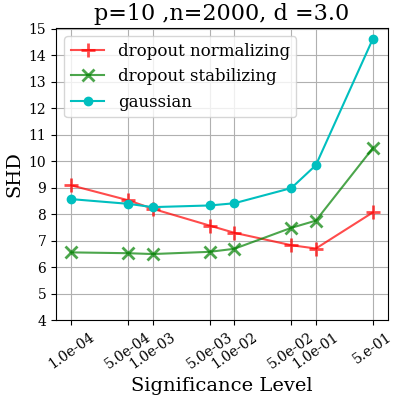}}
	\subfigure[]{\includegraphics[width=.24\linewidth]{figures/SHDSKELp10n10000d3.png}} 
	\subfigure[]{\includegraphics[width=.24\linewidth]{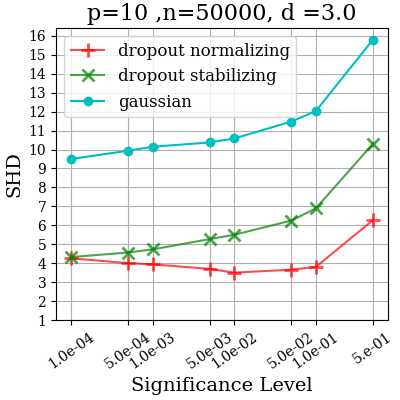}} \\
	\subfigure[]{\includegraphics[width=.24\linewidth]{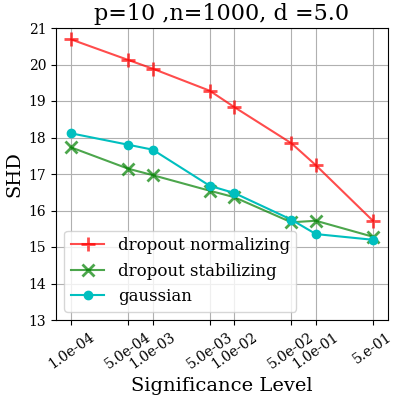}}
	\subfigure[]{\includegraphics[width=.24\linewidth]{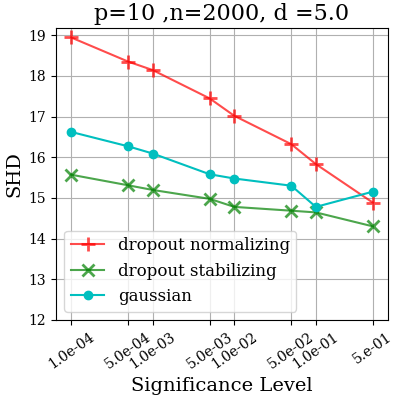}}
	\subfigure[]{\includegraphics[width=.24\linewidth]{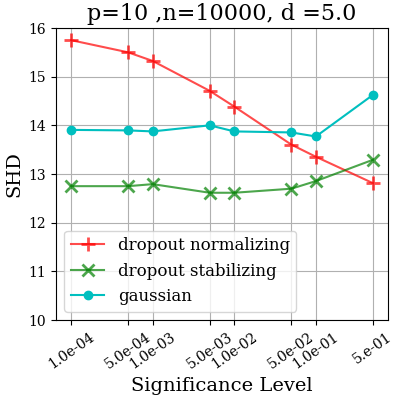}} 
	\subfigure[]{\includegraphics[width=.24\linewidth]{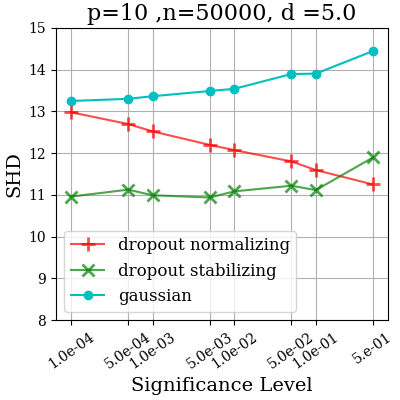}} 
	\caption{SHD for evaluating the estimated skeleton of the true DAG using dropout stabilizing transform, dropout normalizing transform, and Gaussian CI test in simulations with $p=10$ and $n \in \{1000, 2000, 10000, 50000\}$ and $d \in \{3, 5\}$.}
\end{figure}

\begin{figure}[H]
	\centering
	\vspace{-0.3cm}
	\subfigure[]{\includegraphics[width=.24\linewidth]{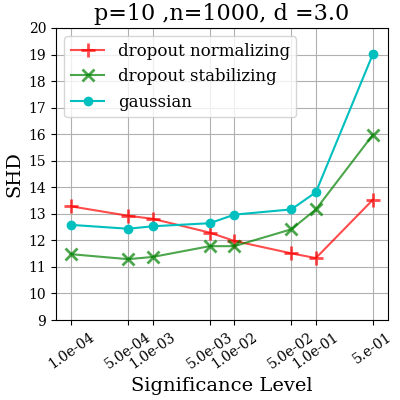}}
	\subfigure[]{\includegraphics[width=.24\linewidth]{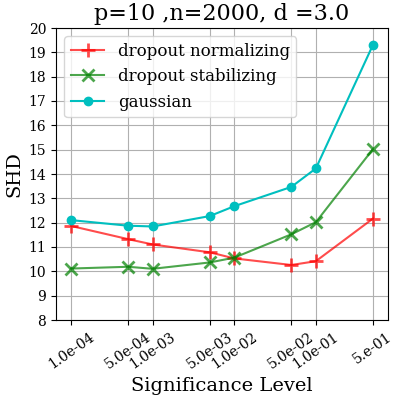}}
	\subfigure[]{\includegraphics[width=.24\linewidth]{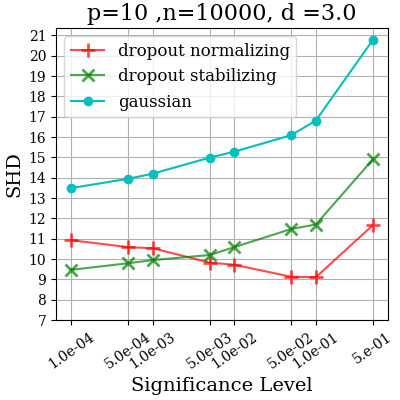}} 
	\subfigure[]{\includegraphics[width=.24\linewidth]{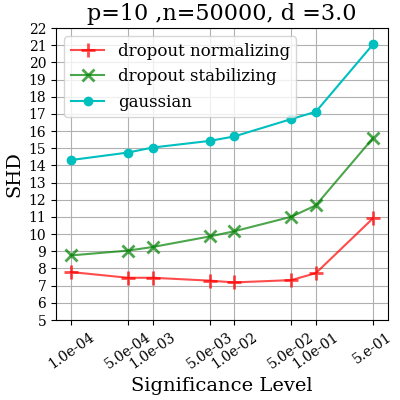}} \\
	\subfigure[]{\includegraphics[width=.24\linewidth]{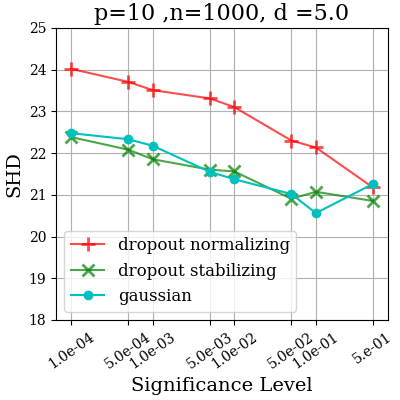}}
	\subfigure[]{\includegraphics[width=.24\linewidth]{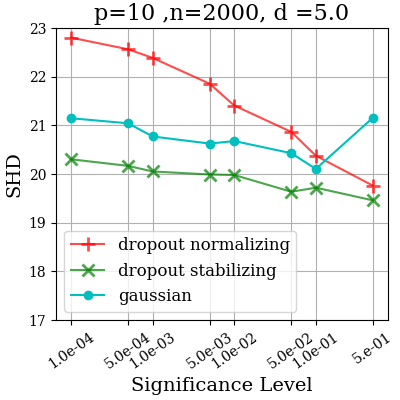}}
	\subfigure[]{\includegraphics[width=.24\linewidth]{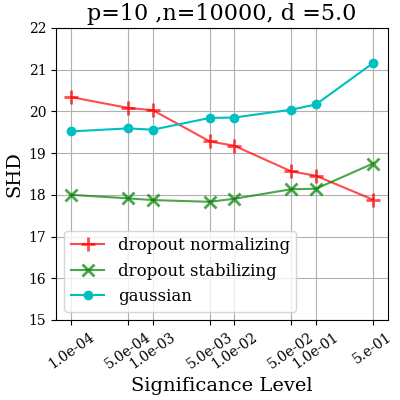}} 
	\subfigure[]{\includegraphics[width=.24\linewidth]{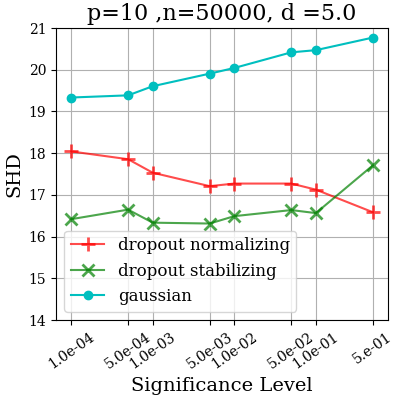}} 
	\caption{SHD for evaluating the estimated CPDAG of the true DAG using dropout stabilizing transform, dropout normalizing transform, and Gaussian CI test in simulations with $p=10$ and $n \in \{1000, 2000, 10000, 50000\}$ and $d \in \{3, 5\}$.}
\end{figure}

\begin{figure}[H]
	\centering
	\includegraphics[width=.7\linewidth]{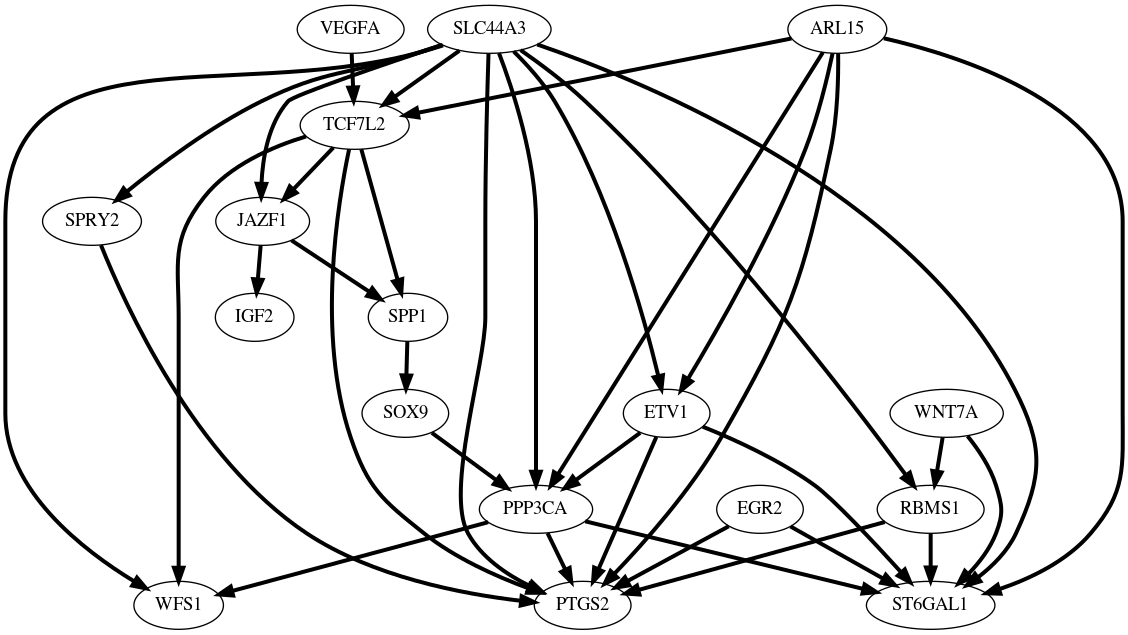}
	\caption{Gene regulatory network inferred from the pancreas data set collected with inDrop. Dropout stabilizing transform was used to learn the causal edges between latent error-free genes.}
\end{figure}

\end{appendices}

\end{document}